\documentclass[journal, a4paper, twoside]{IEEEtran}

\usepackage{amsmath,amssymb,graphicx,epsfig, cite,algorithm,algorithmic,epstopdf, url}
\usepackage[utf8]{inputenc}
\usepackage[mathscr]{euscript}
\usepackage{hyperref}
\usepackage{color}
\usepackage{amsthm}
\usepackage{bbm}
\usepackage{booktabs}

\usepackage{algorithm}
\usepackage{algorithmic}

\usepackage{siunitx}
\def\minwrt[#1]{\underset{#1}{\text{minimize }}}
\def\argminwrt[#1]{\underset{#1}{\text{arg min }}}
\def\maxwrt[#1]{\underset{#1}{\text{maximize }}}
\def\maxemphwrt[#1]{\underset{#1}{\text{\emph{maximize} }}}

\newtheorem{remark}{Remark}
\newtheorem{proposition}{Proposition}
\newtheorem{lemma}{Lemma}
\newtheorem{definition}{Definition}

\newcommand{\norm}[1]{\left\lVert#1\right\rVert}
\newcommand{\abs}[1]{\left|#1\right|}

\def\bx{{\bf x}}
\def\by{{\bf y}}

\def\ccM{{\mathcal{M}}}
\def\ccI{{\mathcal{I}}}

\def\RC{{\mathbb{C}}}
\def\RN{{\mathbb{N}}}
\def\RR{{\mathbb{R}}}
\def\RT{{\mathbb{T}}}

\newcommand{\bomega}{\boldsymbol{\omega}}

\newcommand{\Omtspec}{S}

\newcommand{\costfunc}{c}

\newcommand{\expop}{\mathbb{E}}

\newcommand{\harmspec}{\Phi}

\newcommand{\misparam}{\theta}

\newcommand{\inharm}{\Delta} 
\newcommand{\missignal}{\mu} 
\newcommand{\signal}{x} 
\newcommand{\misnoise}{w} 
\newcommand{\noise}{e} 
\newcommand{\wavediff}{\xi} 
\newcommand{\pitch}{\omega} 
\newcommand{\truefreq}{\tilde{\omega}} 
\newcommand{\amp}{r} 
\newcommand{\trueamp}{\tilde{r}} 
\newcommand{\phase}{\phi} 
\newcommand{\truephase}{\tilde{\phi}} 
\newcommand{\var}{\sigma^2} 
\newcommand{\truevar}{\tilde{\sigma}^2} 
\newcommand{\misscale}{\alpha}

\newcommand{\phasediff}{\breve{\phi}}
\newcommand{\freqdiff}{\breve{\omega}}

\newcommand{\fim}{F}
\newcommand{\extrahessian}{\tilde{F}}
\newcommand{\real}{\mathfrak{Re}}
\newcommand{\imag}{\mathfrak{Im}}
\newcommand{\rreal}{\mathfrak{R}}
\newcommand{\iimag}{\mathfrak{I}}
\newcommand{\firstelement}{\eta}

\newcommand{\mcrlb}{\mathrm{MCRLB}}

\newcommand{\mapcriterion}{\psi_{\mathrm{ML/MAP}}}

\newcommand{\inharmvar}{\sigma^2_\Delta}
\newcommand{\paraminharm}{\breve{\theta}}

\def\argmaxwrt[#1]{\underset{#1}{\text{arg max }}}

\newcommand{\harmsignal}{\mu} 

        \makeatletter
        \def\fps@eqnfloat{!t}
        \def\ftype@eqnfloat{4}
        
        \newenvironment{eqnfloat*}
               {\@dblfloat{eqnfloat}}
               {\end@dblfloat}
        \makeatother
\renewenvironment{thebibliography}[1]{%
 \begin{oldthebibliography}{#1}%
 \setlength{\parskip}{0ex}%
 \setlength{\itemsep}{0ex}%
}%
{%
\end{oldthebibliography}%
}%

\title{Defining Fundamental Frequency for \\Almost Harmonic Signals}
\author{Filip Elvander and  Andreas Jakobsson
\thanks{F.~Elvander is with the Stadius Center for Dynamical Systems, Signal Processing and Data Analytics, KU Leuven, 3001 Leuven, Belgium, (email:~firstname.lastname@esat.kuleuven.be).

A.~Jakobsson is with the Department of Mathematical Statistics, Center for Mathematical Sciences,
 Lund University, SE-22100 Lund, Sweden, \mbox{(email:~firstname.lastname@matstat.lu.se).}}\thanks{\noindent Parts of the material herein has been published in the proceedings of ICASSP~2020.}}

\begin{document}
\maketitle

\begin{abstract}
In this work, we consider the modeling of signals that are almost, but not quite, harmonic, i.e., composed of sinusoids whose frequencies are close to being integer multiples of a common frequency. Typically, in applications, such signals are treated as perfectly harmonic, allowing for the estimation of their fundamental frequency, despite the signals not actually being periodic. Herein, we provide three different definitions of a concept of fundamental frequency for such inharmonic signals and study the implications of the different choices for modeling and estimation. We show that one of the definitions corresponds to a misspecified modeling scenario, and provides a theoretical benchmark for analyzing the behavior of estimators derived under a perfectly harmonic assumption. The second definition stems from optimal mass transport theory and yields a robust and easily interpretable concept of fundamental frequency based on the signals' spectral properties. The third definition interprets the inharmonic signal as an observation of a randomly perturbed harmonic signal. This allows for computing a hybrid information theoretical bound on estimation performance, as well as for finding an estimator attaining the bound. The theoretical findings are illustrated using numerical examples.
\end{abstract}

\section{Introduction}
Signals that may be well modeled as superpositions of harmonically related sinusoids appear in a wide variety of fields. Such signals, often referred to as pitches, are, for instance, used in speech processing for modeling the voiced part of human speech \cite{NorholmJC16_24}, in music information retrieval for extracting musical melodies \cite{MullerEKR11_5}, for monitoring and fault detection of industrial machinery \cite{Randall11}, or assessing the state of human diseases \cite{LittleMHSR09_56}. In such applications, the most commonly employed signal feature is the fundamental frequency, or pitch, corresponding to the reciprocal of the signal period. Due to this, considerable effort has been directed towards developing estimators that are statistically, as well as computationally, efficient \cite{NielsenJJCJ17_135}, with more recent contributions extending the estimation problem to signals containing multiple harmonic structures (see, e.g., \cite{ChristensenJ09}). However, in some cases, the assumed harmonic structure is only approximate, i.e., there exists no fundamental frequency such that the frequencies of the sinusoidal are all integer multiples of it, or, equivalently, the signal is not periodic. Such signals are referred to as inharmonic and may be found, e.g., in the sound produced by stringed musical instruments \cite{Fletcher88R} and to some extent the human voice \cite{GeorgeS97_5}. In the former case, there exist parametric models for some instruments based on their physical properties, describing the precise deviation from a perfect harmonic structure \cite{Fletcher88R}. This allows for the formulation of efficient estimators of the signal parameters \cite{ZhangCJM10_18,ButtASJ13_icassp}, as well for deriving information theoretical bounds on estimation performance, such as the Cram\'er-Rao lower bound (CRLB) \cite{Kay93}. However, when no apparent structure of the inharmonicity is known, it is less clear how to efficiently estimate the frequency content of such signals. Intuitively, one should be able to achieve better estimation performance than what is possible under the assumption of an unstructured model, i.e., when there is no relation between the frequencies of the sinusoidal components. Although one may use robust methods \cite{ZhangCJM10_18} or simply resort to applying estimators derived for the perfectly harmonic case, it is not clear what the relevant bounds on estimation performance are. In fact, when applying a harmonic estimator to an inharmonic signal, it is seldom stated what quantity one actually aims to estimate. This is the question this work aims to answer: what does the concept of a pitch mean for inharmonic signals? We provide three possible definitions, emanating from three different views of the inharmonic signal, and study the implications of the different choices.

Firstly, we propose to define the pitch of inharmonic signals via the best $\ell_2$ approximation. We show that for the case of Gaussian additive noise, this corresponds to the best approximation in the Kullback-Leibler sense, allowing for utilizing the framework of misspecified estimation and for interpreting the pitch definition as a pseudo-true fundamental frequency (see, e.g., \cite{FortunatiGGR17_34}). Furthermore, this allows us to use the misspecified CRLB (MCRLB) \cite{RichmondH15_63} as a bound on estimation performance, which we derive for the inharmonic model. This bound is asymptotically tight and attained by the misspecified maximum likelihood estimator (MLE). Thus, this definition formalizes the inherent assumptions behind applying harmonic estimators to inharmonic signals. However, although practically useful, we show that this definition is problematic for signals with a long time duration, as well as depends on parameters that may be considered nuisance.

As an alternative, we consider approximating the spectral representation of infinite-length versions of the inharmonic signal. Building on the concept of optimal mass transport (OMT) (see, e.g., \cite{Villani08}), we consider the harmonic spectrum closest in the OMT sense, yielding a definition of the fundamental frequency. We show that this definition has some attractive properties, such as stability to small perturbations, as well as an intuitive appeal. Furthermore, when using the definition as a basis for estimating the pitch, the resulting estimator allows for a closed-form expression for the asymptotic variance, and, in the case of perfect harmonicity and white Gaussian noise, has the same asymptotical performance as the MLE. In this case, the resulting estimator corresponds asymptotically to an estimator formed using the extended invariance principle (EXIP) \cite{Stoica1989}, fitting a set of unstructured frequency estimates to a perfectly harmonic structure. The EXIP concept extends the invariance principle for ML estimation to the case when the variable transformation is not bijective and has earlier been successfully applied in, e.g., array processing problems in the presence of calibration errors \cite{SwindleS98_86}, as well as for pitch estimation \cite{LiSL00_80} under the assumption of perfect harmonicity.

Lastly, we consider modeling inharmonic signals within a stochastic framework, wherein the noiseless signal is regarded as an observation of a stochastic process. Specifically, we model deviations from the perfectly harmonic model as zero-mean random variables, allowing for interpreting the pitch as an expectation. As the resulting model contains both deterministic and random parameters, we derive the hybrid CRLB (HCRLB) as a lower bound on estimation mean squared error (MSE). We also derive an easily implementable hybrid ML/maximum a posteriori (MAP) estimator that asymptotically attains the resulting HCRLB.

We compare and contrast the three definitions with each other, highlighting their relative merits and applicability, as well as provide numerical illustrations of the derived estimators and proposed bounds.
%
\section{Almost harmonic signals} \label{sec:almost_harm_signals}
Consider the noiseless signal model\footnote{In the interest of generality, we here consider complex-valued signals. For real-valued signals, a corresponding complex version may easily be formed as the discrete-time analytical signal \cite{Marple99_47}.}
\begin{align}\label{eq:noiseless_sine_model}
	\signal_t =  \sum_{k=1}^K \trueamp_k e^{i\truephase_k + i\truefreq_k t},
\end{align}
for $t = 0,1,\ldots,N-1$, and $N\in \RN$, where $\trueamp_k>0$, $\truephase_k \in [-\pi,\pi)$, and $\truefreq_k \in [-\pi,\pi)$ denote the amplitude, initial phase, and frequency, respectively, for the $k$:th signal component. Then, if there exists $\pitch_0 \in [-\pi,\pi)$ such that $\truefreq_k = k\pitch_0$, for $k = 1,\ldots,K$, the signal in \eqref{eq:noiseless_sine_model} is referred to as being harmonic, with fundamental frequency, or pitch, $\pitch_0$. Note that this is the case also if some of the components are missing, i.e., $\trueamp_k = 0$ for some $k$, as the signal period is still $2\pi/\pitch_0$. The model in \eqref{eq:noiseless_sine_model} appears in many signal processing applications, not least in audio processing, and considerable effort has been directed towards deriving estimators for the parameters $\pitch_0$ and $(\trueamp_k,\truephase_k)$, for $k=1,\ldots,K$ (see, e.g., \cite{ChristensenJ09} for an overview). However, for inharmonic signals, the integer relationship between the component frequencies is only approximate \cite{Klapuri03_11,GeorgeS97_5}. That is,
\begin{align} \label{eq:inharm_freqs}
	\truefreq_k = k\pitch_0 + \inharm_k, \; k = 1,\ldots,K,
\end{align}
where $\inharm_k$ are called inharmonicity parameters. As to distinguish from the case with completely unrelated sinusoidal components, it is here assumed that the inharmonicity parameters are small in the sense that \mbox{$\abs{\inharm_k} \ll \pitch_0$}. Thus, assuming that the component frequencies satisfy \eqref{eq:inharm_freqs} constitutes a type of middle ground between the highly structured harmonic model, where $\inharm_k = 0$, for all $k$, and the unstructured sinusoidal model, where there is no relation between the frequency components. In some special cases, parametric models for the inharmonicity exists. For example, a common model for vibrating strings is
\begin{align} \label{eq:string_model}
	\truefreq_k = k\omega_0 \sqrt{1+k^2\beta}
\end{align}
where $\beta > 0 $ is a parameter related to the stiffness of the string \cite{Fletcher88R}. Thus, for this model, the inharmonicity parameters are given by $\inharm_k = k\omega_0 \left(\sqrt{1+k^2\beta}-1\right) > 0$, for all $k$. However, such a structured model may not be assumed in the general case, and, in this work, we therefore do not assume any particular structure of $\left\{ \inharm_k \right\}_{k=1}^K$, or indeed that any useful such structure exists. Rather, we aim to put the intuitive concept of inharmonic, i.e., almost harmonic, signals on more solid foundation by offering three conceptually different definitions of the meaning of a fundamental frequency for a non-periodic signal. Common for the three definitions will be the existence of a perfectly harmonic waveform, 
\begin{align} \label{eq:harm_waveform}
	\harmsignal_t = \sum_{\ell=1}^L \amp_\ell e^{i\phase_\ell +i \pitch_0\ell t},
\end{align}
with $L$ not necessarily being equal to $K$, that is a best approximation of $x_t$, with the definition of optimality differing between the definitions. In this framework, $\pitch_0$ will be the definition of the pitch for an inharmonic signal $\signal_t$. Furthermore, in order to find bounds on estimation performance, as well as to formulate estimators, we assume that the measured signal is well modeled as
\begin{align} \label{eq:measurement_equation}
	y_t = \signal_t +  \noise_t,
\end{align}
where $\noise_t$ is a circularly symmetric white Gaussian noise with variance $\truevar$. However, most results presented herein may be readily extended to non-white noise processes. Throughout, let
\begin{align*}
	\misparam = \left[\begin{array}{ccccccc} \pitch_0 & \phase_1 & \ldots & \phase_L & \amp_1 & \ldots & \amp_L\end{array}\right]^T,
\end{align*}
denote the parameter vector defining the approximating signal, i.e., $\harmsignal_t \equiv \harmsignal_t(\misparam)$, and let $\by = \left[\begin{array}{ccccccc} y_0 & \ldots & y_{N-1} \end{array}\right]^T$ be the vector of available signal samples. For ease of reference, we recall that the probability density function (pdf) of $\by$ is thus
\begin{align*}
	p_{Y}(\by) = \frac{1}{(\pi \truevar)^N}\exp\left( -\frac{1}{\truevar} \sum_{t=0}^{N-1}\abs{y_t - x_t}^2 \right).
\end{align*}
%
%
\section{$\ell_2$ optimality and misspecified models} \label{sec:l2_def}
We initially consider approximating the inharmonic signal in~\eqref{eq:noiseless_sine_model} in the $\ell_2$ sense, i.e.,
\begin{align} \label{eq:ls_approximation}
	\misparam_0 = \argminwrt[\misparam] \frac{1}{N}\sum_{t=0}^{N-1} \abs{x_t-\harmsignal_t(\misparam)}^2.
\end{align}
That is, the approximating pitch is the harmonic waveform yielding the least squared deviation from the inharmonic signal. For notational brevity, $L = K$ for the $\ell_2$ approximation. Approximations such as \eqref{eq:ls_approximation} has earlier been applied in speech coding applications for decreasing the data rate in speech analysis/synthesis systems \cite{McAulayQ90_icassp}.
In addition to the intuitive appeal of this choice, as well as the tractability of computing $\misparam_0$, this approximation may be interpreted as a so-called pseudo-true parameter within the framework of misspecified models. Specifically, consider a scenario in which it is (erroneously) believed that the signal samples $y_t$ are perfectly harmonic, i.e., generated as $y_t = \harmsignal_t + \misnoise_t$, where $\misnoise_t$ is a circularly symmetric white Gaussian noise with (unknown) variance $\var$. It may here be noted that the noise processes $\noise_t$ and $\misnoise_t$ are not necessarily equal in distribution as it is allowed that $\var \neq \truevar$. With this, one may consider the following definition.
%
%
\begin{definition}[Pseudo-true parameter\cite{FortunatiGGR17_34}] 
	Consider a signal sample $\by$. For a pdf $p$, parametrized by the parameter vector $\misparam$, the pseudo-true parameter, $\theta_0$, is defined as
	\begin{align} \label{eq:pseudo_true}
		\misparam_0 = \argminwrt[\misparam] -\mathbb{E}_{\by}\left( \log {p}(\by;\misparam)  \right),
	\end{align}
where $\mathbb{E}_{\by}$ denotes expectation with respect to the pdf of $\by$.
\end{definition}
%
As may be noted, the minimization criterion in \eqref{eq:pseudo_true} is, up to an additive constant not depending on the parameter $\misparam$, equal to the Kullback-Leibler divergence between the pdf of the assumed model and the actual pdf of the signal sample. That is, the pseudo-true parameter $\misparam_0$ realizes the best Kullback-Leibler approximation of the pdf of the signal within the parametric family $p$. In our case, the following proposition holds.
%
%
\begin{proposition}[Pseudo-true parameter] \label{prop:pseudo_true_param}
Under the Gaussian assumption, the pseudo-true parameter $\theta_0$ for the harmonic model is given by \eqref{eq:ls_approximation}, and the pseudo-true variance parameter for the additive noise is
\begin{align*}
	\var = \truevar + \frac{1}{N} \sum_{t=0}^{N-1} \abs{\wavediff_t(\misparam_0)}^2,
\end{align*}
where $\wavediff_t(\misparam) \triangleq \missignal_t(\misparam) - x_t$.
\end{proposition}
\begin{proof}
As both the assumed and true distributions are circularly symmetric white Gaussian, the result follows directly.
\end{proof}
%
%
%
%
From this, we may conclude that approximating the inharmonic signal in $\ell_2$ may be interpreted as finding the Gaussian pdf with mean identical to a periodic waveform that best approximates the true signal pdf in the Kullback-Leibler sense. It may here be noted that the variance $\var$ of the noise $\misnoise_t$ in the harmonic model is potentially greater than that of the inharmonic model due to the imperfect fit of the waveforms. However, as may be readily verified, the value of the pseudo-true parameter $\misparam_0$ does not depend on neither $\var$ nor $\truevar$. Furthermore, it can be shown that the misspecified MLE (MMLE), i.e., the MLE derived under an assumed model different from that of the actual measurements, asymptotically tends to the pseudo-true parameter, i.e., $\hat{\misparam}_{MMLE} \to \misparam_0$ almost surely as the signal to noise ratio (SNR), or number of signal samples, $N$, tends to infinity \cite{FortunatiGGR17_34}. This leads to a very practical consequence: the harmonic waveform $\harmsignal_t$ resulting from the $\ell_2$ approximation in \eqref{eq:ls_approximation} corresponds to the expected result when applying an harmonic MLE, or approximations thereof \cite{ChristensenJ09}, to inharmonic measurements. We summarize this in the following definition.
\begin{definition} \label{def:l2_pitch}
	Let $x_t$, for $t = 0,\ldots,N-1$, be an inharmonic waveform. Then, the best harmonic approximation in $\ell_2$ is given by $\harmsignal_t(\misparam_0)$, where $\misparam_0$ solves \eqref{eq:ls_approximation}.
\end{definition}
The harmonic signal in Def.~\ref{def:l2_pitch} may be seen as the quantity being (tacitly) estimated when applying estimators derived under an harmonic assumption to inharmonic signals. 
Furthermore, by considering the interpretation as the pseudo-true parameter, one may find a bound on performance on any unbiased estimator of $\misparam_0$. Such a family of bounds is the misspecified CRLB (MCRLB) \cite{RichmondH15_63}. Specifically, considering estimators that satisfy the MLE unbiasedness conditions, the following proposition, adapted from \cite{RichmondH15_63}, yields a bound on estimator variance.
%
%
%
\begin{proposition} \label{th:mrclb}
Let $\hat{\misparam}$ be an estimator of $\misparam_0$ that is unbiased under the signal pdf. Then,
\begin{align} \label{eq:mrclb}
	\expop_\by\left( (\hat{\misparam}-\misparam_0)(\hat{\misparam}-\misparam_0)^T  \right) \succeq A(\misparam_0)^{-1}\fim(\misparam_0) A(\misparam_0)^{-1}
\end{align}
where $A(\misparam) = -\frac{\var}{\truevar}\fim(\misparam) - \extrahessian(\misparam)$ and
\begin{align*}
	\fim(\misparam) \!&=\!\frac{2\truevar}{(\var)^2}\!\sum_{t=0}^{N-1}\!\nabla_\misparam \missignal_t^\rreal(\misparam) \nabla_\misparam \missignal_t^\rreal(\misparam)^T\!+\!\nabla_\misparam \missignal_t^\iimag(\misparam) \nabla_\misparam \missignal_t^\iimag(\misparam)^T,\\
	\extrahessian(\misparam)\!&= \!\frac{2}{\var}\sum_{t=0}^{N-1} \left(\wavediff_t^\rreal(\misparam) \nabla^2_\misparam \missignal_t^\rreal(\misparam)  + \wavediff_t^\iimag(\misparam) \nabla^2_\misparam \missignal_t^\iimag(\misparam) \right).
\end{align*}

%
Here, $(\cdot)^\rreal = \real(\cdot)$ and $(\cdot)^\iimag = \imag(\cdot)$ denote the real and imaginary parts, respectively. 
\end{proposition}
\begin{proof}
The result follows directly from \cite{RichmondH15_63}.
\end{proof}
%
%
%
%
The MCRLB is given by the diagonal of the right-hand side of \eqref{eq:mrclb} and thus provides a lower bound on the variance of any estimator of $\theta_0$ that is unbiased under the pdf of the measured signal. It may be noted that when $\inharm_k = 0 $, for all $k$, i.e., when the signal $\signal_t$ is perfectly harmonic, $\extrahessian(\misparam_0) = 0$, and $\fim(\misparam_0)$ is the standard Fisher information matrix (FIM). In this case, the MCRLB coincides with the CRLB of a harmonic signal. Thus, the MCRLB provides a means of assessing the performance of any estimator derived under the harmonic assumption, even when the observed signal is inharmonic. Further, as $N \to \infty$, one may express the MCRLB corresponding to the pseudo-true fundamental frequency in closed form, 
as detailed below.
%
\begin{proposition}[Asymptotic MCRLB] \label{prop:asymp_mcrlb}
Let the pseudo-true parameter be
\begin{align*}
	\misparam_0 = \left[\begin{array}{ccccccc} \pitch_0 & \phase_1 & \ldots & \phase_K & \amp_1 & \ldots & \amp_K\end{array}\right]^T.
\end{align*}
Then, as $N\to \infty$, the asymptotic $\mcrlb$ for the pseudo-true fundamental frequency $\pitch_0$ is given by
\begin{align}
	\mcrlb(\pitch_0) = \truevar \frac{C + E}{\left(C-E + Z + D   \right)^2} + \mathcal{O}(N^{-4}),
\end{align}
where $C = \frac{N(N^2-1)\sum_{k=1}^K k^2 \amp_k^2}{6}$, and
\begin{align*}
	Z &= -2 \sum_{k=1}^K k^2\amp_k^2 \frac{N(N-1)(2N-1)}{6} \\
	&\quad + 2 \sum_{k=1}^K \sum_{t=0}^{N-1} k^2\amp_k\trueamp_k t^2 \cos(\phasediff_k+\freqdiff_kt) \\
	\!D\!&=\!2(N\!-\!1)\!\Bigg[\!\frac{N(\!N\!-\!1\!)}{2}\!\sum_{k=1}^K\!k^2\amp_k^2\!-\!\sum_{k=1}^K\!\sum_{t=0}^{N-1} \!k^2\amp_k\trueamp_kt\cos(\phasediff_k\!+\!\freqdiff_k t)\!\Bigg]\\
	E &= \frac{2}{N}\sum_{k=1}^K k^2 \trueamp_k^2 \left( \sum_{t=0}^{N-1} t \sin(\phasediff_k+\freqdiff_kt) \right)^2 \\ 
	&\quad+ \frac{2}{N}\sum_{k=1}^K k^2 \left( \trueamp_k \sum_{t=0}^{N-1} t \cos(\phasediff_k+\freqdiff_kt)  - \amp_k\frac{N(N-1)}{2}\right)^2
\end{align*}
where $\phasediff_k = \phase_k - \truephase_k$ and $\freqdiff_k = k\pitch_0 - \truefreq_k$, for $k = 1,\ldots,K$.
Here, $\mathcal{O}(N^{-4})$ denotes the order of the error term resulting from the asymptotic approximation.
\end{proposition}
\begin{proof}
See the appendix.
\end{proof}
It may be noted that in the perfectly harmonic case, the terms $E, Z$, and $D$ are all equal to zero as the pseudo-true parameter $\misparam_0$ then coincides with the actual signal parameter. It is worth noting that $\truevar/C$ is equal to the asymptotic CRLB for the perfectly harmonic model \cite{ChristensenJJ07_15}.
%
%
\begin{figure}[t]
        \centering
            \includegraphics[width=.46\textwidth]{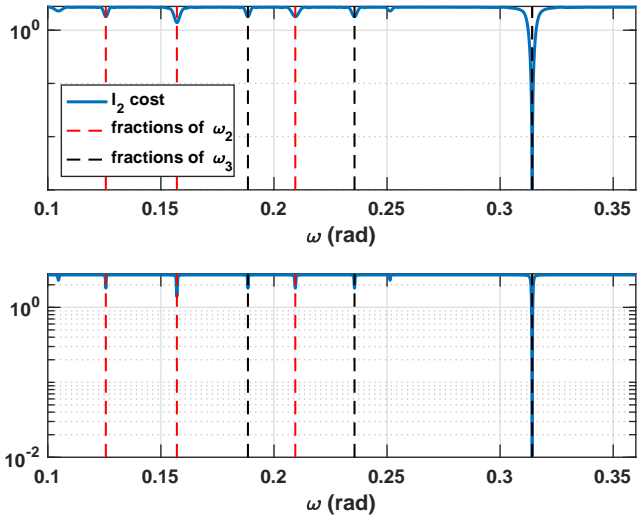}
           \caption{The $\ell_2$ cost function used in Def.~\ref{def:l2_pitch} when applied to a perfectly harmonic signal with fundamental frequency $\pitch_0 = \pi/10$ for two different sample lengths. Top panel: $N = 500$. Bottom panel: $N = 3000$.} 
            \label{fig:l2_cost_combined_harm}
\vspace{-2mm}\end{figure}
%
%
\begin{remark}
As one may compute the MCRLB for an estimate of the pseudo-true fundamental frequency $\pitch_0$, it is also possible to construct a bound for the expected MSE for misspecified estimators. That is, if one considers an estimate of $\pitch_0$ to be a misspecified estimate of $\truefreq_1$, the theoretical MSE is given by
\begin{align}
	\expop_{\by}\left( (\hat{\pitch}_0- \truefreq_1)^2\right) = \mcrlb(\pitch_0) + (\pitch_0-\truefreq_1)^2.
\end{align}
In fact, as will be illustrated in the numerical section, the MSE for this misspecified estimate may for moderate values of SNR and sample lengths $N$ be lower than the CRLB for an unstructured sinusoidal model, due to the MCRLB often being considerable smaller than the CRLB. 
\end{remark}
%
\begin{figure}[t]
        \centering
            \includegraphics[width=.46\textwidth]{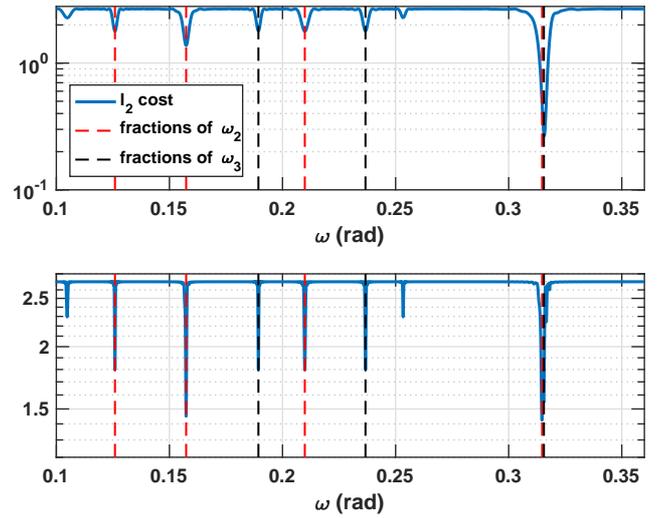}
           \caption{The $\ell_2$ cost function used in Def.~\ref{def:l2_pitch} when applied to an inharmonic signal generated using \eqref{eq:string_model}, with $\pitch_0 = \pi/10$ and $\beta = 10^{-3}$ for two different sample lengths. Top panel: $N = 500$. Bottom panel: $N = 3000$.} 
            \label{fig:l2_cost_combined_inharm}
\vspace{-2mm}\end{figure}
%
Despite its appeal as formalizing tacit assumptions behind using harmonic estimators, as well as resulting in easily computed performance bounds, Def.~\ref{def:l2_pitch} becomes unsatisfactory if one allows for very long signals. Although the normalization by $N$ allows for considering limits of the criterion \eqref{eq:ls_approximation}, as it is guaranteed to be finite, the definition of pitch becomes ambiguous. To see this, it may be noted that the correlation of any two sinusoids with distinct frequencies tends to zero as $N\to \infty$. Thus, \eqref{eq:ls_approximation} is minimized by setting a multiple of the pseudo-true $\pitch_0$ equal to the frequency $\truefreq_k$ corresponding to the largest amplitude $\trueamp_k$. However, any integer multiple is an equally valid choice. Def.~\ref{def:l2_pitch} allows for selecting $\pitch_0$ as $\pitch_0 = \truefreq_k/\ell$, where $k = \text{arg max}_m \trueamp_m$, for any $\ell \in \left\{ 1,\ldots,K \right\}$. Thus, $K$ indistinguishable candidates for $\pitch_0$ exist if $N$ is large enough. It may also be noted that $\pitch_0$ depends on the initial phases $\truephase_k$ of the inharmonic signal, parameters that may be considered nuisance.

The issue of ambiguity for large $N$ is illustrated in Figures~\ref{fig:l2_cost_combined_harm} and \ref{fig:l2_cost_combined_inharm}. Specifically, Figure~\ref{fig:l2_cost_combined_harm} displays the $\ell_2$ cost function for a perfectly harmonic signal with $\pitch_0 = \pi/10$ consisting of five harmonics, where $\trueamp_2 = \trueamp_3$ are the largest amplitudes, with the top and bottom panels of Figure~\ref{fig:l2_cost_combined_harm} showing the cost for $N = 500$ and $N = 3000$, respectively. As can be seen, for both cases, $\pitch_0$ corresponds to the unique global minimum, with a cost of exactly zero. In contrast, Figure~\ref{fig:l2_cost_combined_inharm} displays the same scenario, with the difference being that the perfectly harmonic structure having been replaced with the string model in \eqref{eq:string_model} with $\beta = 10^{-3}$. For $N = 500$, the global minimum is still unique, and the definition thus unambiguous. However, for $N = 3000$, the cost function value at the local minimum at $\truefreq_2/4$ approaches that of the global minimum. If one lets $N\to \infty$, these cost functions values will become identical, in addition to several other isolated global minima appearing.
It is worth noting that, as the sample length becomes longer, the larger the size of the error incurred by the harmonic approximation is due to the sinusoidal components of the approximating harmonic signal becoming increasingly orthogonal to the components of the inharmonic signal. Thus, as the sample length approaches infinity, the squared $\ell_2$-difference will approach the sum of the (squared) norms of any sinusoidal component of the inharmonic whose frequency is not matched perfectly by a component of the harmonic approximation.
The issue of ambiguity may be addressed by instead considering the spectral properties of the signal $\signal_t$, allowing for defining a harmonic approximation of the more abstract signal, i.e., when no particular sample length $N$ has been specified. This can be achieved through the use of OMT, as described next.
%
%
%
%
\section{An OMT-based definition of pitch}\label{sec:omt_def}
As an alternative to defining the harmonic counterpart of an inharmonic signals by means of waveform approximation, such as in Def.~\ref{def:l2_pitch}, one may instead consider the signals' spectral representation. To that end, assume that the phase parameters $\truephase_k$ in \eqref{eq:noiseless_sine_model} are independent random variables uniformly distributed on $[-\pi,\pi)$, implying that the signal in \eqref{eq:noiseless_sine_model} is a wide-sense stationary process. Then, the spectrum of the signal in \eqref{eq:noiseless_sine_model} is given by
\begin{align}
	\Phi_x(\omega) = 2\pi\sum_{k=1}^K \trueamp_k^2 \delta(\omega-\truefreq_k),
\end{align}
where $\omega \in [-\pi,\pi)$, and $\delta(\cdot)$ denotes the Dirac delta function. This constitutes a more abstract representation of $x_t$ as it is not related to any particular sample length $N$. Furthermore, any harmonic spectrum may be represented as
\begin{align} \label{eq:harmspec}
	\Phi_{\mu}(\omega) = 2\pi\sum_{\ell=1}^L \amp_\ell^2 \delta(\omega-\ell\pitch_0),
\end{align}
where we let $L$ be potentially different from $K$. Thus, one may define the fundamental frequency of an inharmonic signal as the $\pitch_0$ corresponding to an approximating harmonic spectrum $\Phi_{\mu}$. Note however that considering the approximation in $L_p$, for some $p\geq 1$, leads to the same problem as was encountered in the limiting case for $\ell_2$ approximation of the waveform, namely that of selecting $\pitch_0$ as $\pitch_0 = \truefreq_k/\ell$, where $k = \text{arg max}_m \trueamp_m$, for any $\ell \in \left\{ 1,\ldots,L \right\}$ yield equally accurate approximations. However, a measure of distance not related to point-wise comparison of $\Phi_x$ and $\Phi_\mu$ may be obtained by considering the framework of OMT. Within this framework, the best approximation corresponds to the one requiring the least costly perturbation in order to shift the observed signal to that of a perfectly harmonic model, with perturbations being realized by moving the whole distribution, as opposed to the point-wise changes implied by $L_p$ norms.

To formulate this, let $\RT = [-\pi,\pi)$ and let $\ccM_+(\RT)$ denote the set of non-negative, generalized integrable functions on $\RT$. Elements of $\ccM_+(\RT)$ may be interpreted as distributions of mass on $\RT$, and, in particular, $\Phi_x \in \ccM_+(\RT)$ and $\Phi_\mu \in \ccM_+(\RT)$. Then, for any $\Phi_0, \Phi_1 \in \ccM_+(\RT)$ with the same total mass, i.e., $\int_{\RT} \Phi_0(\omega)d\omega = \int_{\RT} \Phi_1(\omega)d\omega$, one may define a notion of distance, $\Omtspec: \ccM_+(\RT)\times \ccM_+(\RT) \to \RR$, between them by the Monge-Kantorovich problem of OMT \cite{Villani08},
\begin{align*}
	\Omtspec(\Phi_0,\Phi_1) &= \min_{\ccM_+(\RT\times\RT)} \int_{\RT\times \RT} M(\omega_1,\omega_2)\costfunc(\omega_1,\omega_2)d\omega_1d\omega_2 \\
	&\quad\text{s.t. } \int_{\RT} M(\cdot,\omega)d\omega = \Phi_0 \;,\; \int_{\RT} M(\omega,\cdot)d\omega = \Phi_1 \nonumber
\end{align*}
where $\costfunc: \RT\times \RT \to \RR_+$ is a cost function defining the cost of moving a unit mass. Here, $M$ is referred to as a transport plan as it may be interpreted as describing how mass is moved from $\Phi_0$ to $\Phi_1$. As the constraints ensure that $M$ transports all available mass, and no other, between $\Phi_0$ and $\Phi_1$, the objective corresponds to the total cost of transport. The idea of using OMT as a measure of distance between spectra has earlier been considered in \cite{GeorgiouKT09_57}, wherein it was shown that $\Omtspec$, for certain choices of $\costfunc$, may be used for defining a metric on $\ccM_+(\RT)$ (see also \cite{ElvanderJK18_66,ElvanderHJK20_171} for corresponding covariance based signal distances). Herein, we will use $\costfunc(\omega_1,\omega_2) = (\omega_1-\omega_2)^2$, i.e., the cost of transport between two frequencies is equal to their squared Euclidean distance. With this, one may find the best harmonic approximation of $\Phi_x$ in the OMT sense by minimizing $\Omtspec(\cdot,\Phi_x)$ over the set of harmonic spectra. Note that, in contrast to $L_p$ minimization, this corresponds to finding the most efficient way of perturbing the spectral peaks of $\Phi_x$ in frequency so that the resulting spectrum is harmonic, with the cost of moving a peak being proportional to its power. Formally,
\begin{align}\label{eq:omt_harm}
	\Phi_\mu = \argminwrt[\Phi \in \Omega_L] \Omtspec(\Phi,\Phi_x)
\end{align}
where $\Omega_L$ is the set of harmonic spectra, i.e.,
\begin{align*}
	\Omega_L =& \Big\{ \Phi \in \ccM_+(\RT) \mid\\
	&
	 \qquad\harmspec(x) = 2\pi\sum_{\ell=1}^L \amp_\ell^2 \delta(\omega-\ell\pitch_0) \;,\; \amp_\ell \geq0, \pitch_0 \in \RT  \Big\}.
\end{align*}
To see that \eqref{eq:omt_harm} may be solved efficiently, note that for any candidate $\pitch_0$, corresponding to a subset of $\Omega_L$, all power at the frequency $\truefreq_k$ in $\Phi_x$ will be transported to the nearest integer multiple of $\pitch_0$ when evaluating $\Omtspec$. That is, letting $\Omega_{L,\pitch_0}$ denote such a subset, i.e., all harmonic spectra with fundamental frequency $\pitch_0$, we get
\begin{align*}
	 \min_{\harmspec \in \Omega_{L,\pitch_0}} \Omtspec(\Phi,\Phi_x)  = 2\pi\sum_{k=1}^K \trueamp_k^2 \min_{\ell \in \left\{ 1,2,\ldots,L \right\}} (\ell\pitch_0 - \truefreq_k)^2.
\end{align*}
Thus, solving \eqref{eq:omt_harm} is equivalent to solving
\begin{align} \label{eq:best_omt_f0}
	\minwrt[\omega_0] 2\pi\sum_{k=1}^K \trueamp_k^2 \min_{\ell \in \left\{ 1,2,\ldots,L \right\}} (\ell\pitch_0 - \truefreq_k)^2.
\end{align}
Furthermore, at least one harmonic spectrum attaining the minimal cost exists, and is given by
\begin{align*}
	\Phi_\mu(\omega) = 2\pi\sum_{\ell = 1}^L \left(\sum_{k \in \ccI_\ell } \trueamp_k^2\right) \delta(\omega-\ell \pitch_0),
\end{align*}
where $\pitch_0$ solves \eqref{eq:best_omt_f0}, and with
\begin{align} \label{eq:harm_index}
	\ccI_\ell = \left\{ k \mid \ell = \argminwrt[m] (m\pitch_0 - \truefreq_k)^2 \right\}
\end{align}
denoting the set of indices that are transported to harmonic $\ell$, for $\ell = 1,\ldots,L$. As may be noted, the maximal harmonic order, $L$, may not be equal to $K$ (indeed, $K$ may be unknown, or may not be a suitable choice for $L$). Although $L$ could be left as a user defined parameter, we here offer a data dependent choice with both intuitive appeal and beneficial practical consequences.
%
%
\begin{definition}[Maximal harmonic order] \label{def:max_harm_order}
For a set $\left\{ \truefreq_k\right\}_{k=1}^K$ such that $\truefreq_{k+1}> \truefreq_{k}$ for $k = 1,\ldots,K-1$, let $d = \min \left\{ \truefreq_1, \truefreq_2 - \truefreq_1, \truefreq_3-\truefreq_2,\ldots, \truefreq_K-\truefreq_{K-1}\right\}$, i.e., the minimum distance between two consecutive sinusoidal components. Then, the harmonic order for the set $\left\{ \truefreq_k\right\}_{k=1}^K$ is defined as $L = \min \left\{ \ell \in \RN \mid \ell d \geq \truefreq_K  \right\}$.
\end{definition}
%
%
\begin{figure}[t]
        \centering
            \includegraphics[width=.46\textwidth]{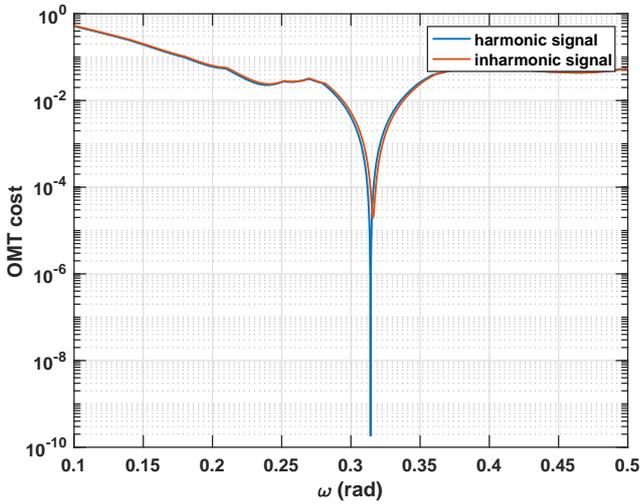}
           \caption{The OMT cost function, $q_L$, for a perfectly harmonic signal with fundamental frequency $\pitch_0 = \pi/10$, as well as for an inharmonic signal generated from the string model in \eqref{eq:string_model} with $\pitch_0 = \pi/10$ and $\beta = 10^{-3}$.} 
            \label{fig:omt_cost_combined_harm_inharm}
\vspace{-2mm}\end{figure}
%
With this, the definition of the OMT harmonic spectrum is:
%
\begin{definition}[Closest harmonic spectrum] \label{def:omt_spectrum}
Let \mbox{$\Phi_x(\omega) = 2\pi\sum_{k=1}^K \trueamp_k^2 \delta(\omega-\truefreq_k)$} be a (possibly) inharmonic spectrum, and let $L$ be the maximal harmonic order as defined in Def.~\ref{def:max_harm_order}. Then, the closest harmonic spectrum (CHS) is defined as
\begin{align*}
	\Phi_\mu(\omega) = 2\pi\sum_{\ell = 1}^L \left(\sum_{k \in \ccI_\ell} \trueamp_k^2\right) \delta(\omega-\ell \pitch_0).
\end{align*}
where $\pitch_0$ solves \eqref{eq:best_omt_f0} and $\ccI_\ell$ is given by \eqref{eq:harm_index}.
\end{definition}
%
It may here be noted that selecting $L$ according to Def.~\ref{def:max_harm_order} acts as a safeguard against the so called sub-octave problem, i.e., associating the spectral lines with $\omega_0/2^P$, for $P\geq 1$, which is bound to happen if $L$ is chosen excessively large.
The following two propositions verify that Def.~\ref{def:omt_spectrum} behaves in a stable and predictable way. To simplify the proofs, define
\begin{align} \label{eq:harm_spectrum_cost}
	q_L(\omega_0) \triangleq 2\pi\sum_{k=1}^K \trueamp_k^2 \min_{\ell \in \left\{ 1,2,\ldots,L \right\}} (\ell\pitch_0 - \truefreq_k)^2,
\end{align}
i.e., the minimal cost of transporting $\Phi_x$ to a harmonic spectrum with fundamental frequency $\pitch_0$ and $L$ harmonics. Thus, the fundamental frequency of the CHS minimizes $q_L$.
Furthermore, it may be noted that the following propositions hold also if one instead of using Def.~\ref{def:max_harm_order} selects $L = K$.
%
%
\begin{proposition}
Let $\mbox{$\Phi_x(\omega) = 2\pi\sum_{k=1}^K \trueamp_k^2 \delta(\omega-k\tilde{\pitch}_0)$}$ be a harmonic spectrum with fundamental frequency $\tilde{\pitch}_0$. Then, the CHS is identical to $\Phi_x$, and, in particular, $\pitch_0 = \tilde{\pitch}_0$.
\end{proposition}
\begin{proof}
Clearly, the maximal harmonic order is $L = K$. Further, $q_K(\tilde{\pitch}_0) = 0$, and $q_K(\omega') > 0$ for any $\omega' \neq \tilde{\pitch}_0$.
\end{proof}
%
This proposition ensures that for a perfectly harmonic spectrum, the CHS is the spectrum itself. As may be seen from \eqref{eq:harm_spectrum_cost}, the function is $q_L$ is non-convex, with several local minima. Thus, for arbitrarily large inharmonicity parameters $\inharm_k$ and arbitrary choices of amplitudes $\trueamp_k$, the fundamental frequency cannot be guaranteed to be found in a certain region (it may be noted that this is also the case for the $\ell_2$ approximation). However, for small harmonic perturbations, the following proposition holds.
%
\begin{figure}[t]
        \centering
            \includegraphics[width=.46\textwidth]{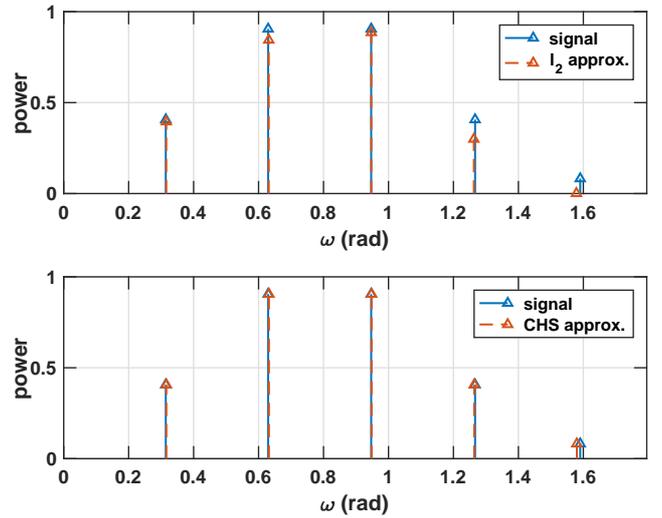}
           \caption{Spectrum of inharmonic signal, as well as spectra of $\ell_2$ and CHS approximations corresponding to Defs.~\ref{def:l2_pitch} and \ref{def:omt_spectrum}, respectively. Top panel: signal and $\ell_2$ approximation. Bottom panel: signal and CHS approximation.} 
            \label{fig:l2_versus_omt_spectrum}
\vspace{-2mm}\end{figure}
%
%
\begin{proposition} \label{prop:bounded_perturbation}
Let $\mbox{$\Phi_x(\omega) = 2\pi\sum_{k=1}^K \trueamp_k^2 \delta(\omega-\truefreq_k)$}$, where $\truefreq_k = k\tilde{\pitch}_0 + \inharm_k$, with $\norm{\inharm}_\infty \triangleq \max_k\abs{\inharm}_k< \tilde{\pitch}_0/(2K+3)$. Then, $L \in \left\{ K, K+1 \right\}$, and the transport cost $q_L$ has a exactly one local minimum $\pitch_0$ on the interval $\left[\tilde{\pitch}_0 - \norm{\inharm}_\infty , \tilde{\pitch}_0 + \norm{\inharm}_\infty\right]$. Furthermore,
\begin{align*}
	\abs{\pitch_0 - \tilde{\pitch}_0} \leq \frac{\sum_{k=1}^K\trueamp_k^2 k}{\sum_{k=1}^K\trueamp_k^2 k^2} \norm{\inharm}_\infty. 
\end{align*}
\end{proposition}
\begin{proof}
See the appendix.
\end{proof}
For non-extreme choices\footnote{Adversarial examples with combinations of very large and small amplitudes can always be constructed as to move the global mimimum.} of the amplitudes $\trueamp_k$, the local minimum in Prop.~\ref{prop:bounded_perturbation} is expected to be also the global minimum, i.e., small inharmonic perturbations is expected to yield a CHS whose pitch is close to that of the unperturbed perfectly harmonic signal. With this, one may conclude that Def.~\ref{def:omt_spectrum} provides a reasonable and well-behaved definition of pitch for inharmonic signals with some advantages over the $\ell_2$ approximation. Firstly, as the approximation is performed in the spectral domain, there is no dependency parameters related to the observation of an instance of the signal, such as sample length or initial phases. Secondly, it codifies the intuitive idea that a signal is almost harmonic if only a slight perturbation of its spectrum is needed to obtain an harmonic structure.

To illustrate the intuitive appeal of Def.~\ref{def:omt_spectrum}, Figure~\ref{fig:omt_cost_combined_harm_inharm} considers the same example as in Figures~\ref{fig:l2_cost_combined_harm} and \ref{fig:l2_cost_combined_inharm}, with the difference being that the approximation is performed in the spectral domain and in the OMT sense instead of in the temporal domain and $\ell_2$. As can be seen, using Def.~\ref{def:omt_spectrum} causes only a small perturbation of the $\pitch_0$ defining the approximating spectrum as compared to the perfectly harmonic case. In fact, the harmonic spectrum of Def.~\ref{def:omt_spectrum} is constructed by slightly shifting the frequency locations of the peaks of the actual signal spectrum. This is illustrated in Figure~\ref{fig:l2_versus_omt_spectrum}, displaying the spectrum\footnote{As the spectra are all singular, the pointmasses are here represented by scaled arrows.} of the inharmonic signal, as well as the spectra of the approximations from Defs.~\ref{def:l2_pitch} and \ref{def:omt_spectrum}. As may be noted, the total power is retained for Def.~\ref{def:omt_spectrum}, whereas amplitudes are underestimated or components altogether missing due to orthogonality for Def.~\ref{def:l2_pitch}. As a side, it may also be noted that the cost function is quadratic in a neighborhood of this fundamental frequency. 

As the $\ell_2$ approximation is closely related to the MLE derived under the perfectly harmonic assumption, Def.~\ref{def:l2_pitch} is readily applicable to actual estimation problems. However, Def.~\ref{def:omt_spectrum} may also be used as a plug-in estimator by replacing the quantities $(\trueamp_k,\truefreq_k)$ by finite-sample estimates $(\hat{\trueamp}_k,\hat{\truefreq}_k)$ and consider the obtained fundamental frequency to be an estimate of $\pitch_0$. The following proposition details the asymptotic behavior of such an estimator.
\begin{proposition}[CHS estimate] \label{prop:chs_estimate}
Let $(\hat{\trueamp}_k,\hat{\truefreq}_k)$ be MLEs of $(\trueamp_k,\truefreq_k)$, for all $k$, obtained from an unstructured sinusoidal model with Gaussian noise. Then, under the assumptions of Prop~\ref{prop:bounded_perturbation}, the plug-in estimate of the CHS pitch, $\hat{\pitch}_0$, is an asymptotically, i.e., as $N\to \infty$, consistent estimator of the CHS $\pitch_0$ with asymptotic variance given by
\begin{align}
	\text{Var}(\hat{\pitch}_0) &= \frac{6\truevar}{N(N^2-1)\sum_{k=1}^K k^2\trueamp_k^2} \label{eq:chs_var}\\
	&\!+\!\frac{2\truevar}{N\left( \sum_{k=1}^K \!k^2\trueamp_k^2  \right)^{\!4}} \!\sum_{k=1}^K k^2\trueamp_k^2 \!\left( \!\sum_{\ell=1}^K \!\ell \trueamp_\ell^2 (\ell\truefreq_k \!-\! k\truefreq_\ell) \!\right)^2\!. \nonumber
\end{align}
\end{proposition}
\begin{proof}
See the appendix.
\end{proof}
\begin{remark}
It may be noted that the first term of \eqref{eq:chs_var} is identical to the asymptotic CRLB for the pitch in a perfectly harmonic model \cite{ChristensenJJ07_15}, whereas the second term is related to deviations of the frequencies from perfect harmonicity. Further, in the perfectly harmonic case, i.e., $\truefreq_k = k\pitch_0$, for all $k$, the second term of \eqref{eq:chs_var} is equal to zero, i.e., the CHS plug-in estimate has the same asymptotic performance as the MLE for the perfectly harmonic model. This is not surprising; in the perfectly harmonic case and under the assumption of additive Gaussian white noise, the criterion $q_L$ minimized by the CHS is asymptotically equivalent to the EXIP cost function corresponding to finding an estimate of the pitch from a set of unstructured frequency estimates \cite{LiSL00_80}. Furthermore, the EXIP estimate has been shown to asymptotically have the same performance as the MLE \cite{Stoica1989}.
\end{remark}
%
%
\section{Stochastic representation}\label{sec:stochastic_def}
As may be noted, the parameters of the inharmonic signal in \eqref{eq:noiseless_sine_model} have in Defs.~\ref{def:l2_pitch} and \ref{def:omt_spectrum} been assumed to be deterministic (although the spectral representation allows for, e.g., the initial phases to be random), and, as a consequence, the corresponding harmonic structure has been defined in terms of approximations. In the third and last definition presented herein, we take an alternative approach and view \eqref{eq:noiseless_sine_model} as a realization of a stochastic process in which the frequency parameters are random. Specifically, for the frequencies $\truefreq_k$ in \eqref{eq:inharm_freqs}, we let $\pitch_0$ be a deterministic parameter, whereas $\inharm_k$ are independent zero-mean random variables so that
\begin{align*}
	\expop\left( \truefreq_k \right) = \expop\left( k\pitch_0 + \inharm_k \right) = k\pitch_0, \;\text{for }k = 1,\ldots,K.
\end{align*}
With this, $\pitch_0$ may be interpreted as the fundamental frequency in an average sense\footnote{Note, however, that the expectation of $x_t$ is not a periodic waveform.}. In order to arrive at a model allowing for manipulations, we will herein further assume that the inharmonicity parameters may be well modeled as Gaussian random variables, i.e., $\inharm_k \in \mathcal{N}(0,\inharmvar)$, where $\inharmvar$ denotes the common variance\footnote{This choice is made to simplify the exposition. However, all results herein may be readily extended to allowing the variance to differ among the inharmonicity parameters.}. It may be noted that this representation constitutes a weaker, less structured assumption on the nature of the inharmonic deviations than that provided by, e.g., the stiff string model in \eqref{eq:string_model}, and may therefore be used to model inharmonic signals for which there is no known parametric description. We summarize this in the following definition.
%

\begin{definition}[Expected harmonic signal] \label{def:stochastic_def}
The signal in \eqref{eq:noiseless_sine_model} is a realization of the stochastic process
\begin{align} \label{eq:stochastic_harmonic_sines}
	\signal_t =  \sum_{k=1}^K \trueamp_k e^{i\truephase_k + i\left( k\pitch_0 + \inharm_k \right) t},
\end{align}
where the pdf of $\inharm = \left[\begin{array}{ccc} \inharm_1 & \ldots & \inharm_K \end{array} \right]^T$ is given by
\begin{align*}
	p(\inharm)&=\frac{1}{(2\pi \inharmvar)^{K/2}}\exp\left(-\frac{1}{2\inharmvar} \norm{\Delta}_2^2  \right).
\end{align*}
The corresponding harmonic signal is given by \eqref{eq:harm_waveform}.
\end{definition}
%
It may be noted that there is an interesting relation between Defs.~\ref{def:omt_spectrum} and \ref{def:stochastic_def}. Specifically, for a given realization of $x_t$, corresponding to an observation of $\truefreq_k$, the waveform may be explained perfectly by matching amplitudes, initial phases, and by making any choice of $\pitch_0$ and $\inharm$ such that
\begin{align*}
	k\pitch_0 + \inharm_k = \truefreq_k, \;\text{for }k =1,\ldots,K.
\end{align*}
However, the choice that is most likely, in the sense of maximizing $p(\inharm)$, is given by
\begin{align*}
	\pitch_0 = \argminwrt[\omega] \sum_{k=1}^K (k\omega - \truefreq_k)^2 = \frac{\sum_{k=1}^K k \truefreq_k}{\sum_{k=1}^K k^2},
\end{align*}
and $\inharm_k = \truefreq_k - \pitch_0/k$, for all $k$. That is, the estimate of $\mu_t$ corresponds to a variation of Def.~\ref{def:omt_spectrum} in which the cost of transport is not related to the component power, i.e., as if $\trueamp_k = 1$, for all $k$. It may also be noted that this model is identifiable for all finite $\inharmvar$, i.e., one is not forced to identify one of the frequencies $\truefreq_k$ with an integer multiple of $\pitch_0$ in order for the model to be well-defined.

Considering the noisy observation model in \eqref{eq:measurement_equation}, we may proceed by asking what type of performance bounds and estimators that are relevant for the signal in Def.~\ref{def:stochastic_def}. To this end, let $\paraminharm = \left[\begin{array}{cc} \theta^T & \inharm^T \end{array}\right]^T$ denote the concatenation of the deterministic vector $\theta$ and the stochastic vector $\inharm$. Furthermore, let $\bx(\paraminharm)$ denote the vector consisting of the $N$ signals samples in \eqref{eq:stochastic_harmonic_sines}, parametrized by $\paraminharm$. Then, assuming that the measurement noise is independent of $\inharm$, the joint pdf of the measurement $\by$ and the inharmonicity $\inharm$ is given by
\begin{align} \label{eq:hybrid_pdf}
	p(\by,\inharm;\theta) = p(\by \mid \inharm;\theta) p(\inharm)
\end{align}
where
\begin{align*}
	p(\by \mid \inharm;\theta) &= \frac{1}{(\pi\truevar)^N} \exp\left(-\frac{1}{\truevar} \norm{\by-\bx(\theta,\inharm)}_2^2  \right)
\end{align*}
is the conditional pdf of the measurement. From this, it may be noted that it is not trivial to compute the CRLB for $\theta$, nor to derive the MLE, as the marginal density $p(\by;\theta)$ is not available; this requires computing a $K$ dimensional integral with a non-linear integrand. Also, the Bayesian CRLB is not applicable in this case as prior distributions are only available for a subset of the parameters. However, it is possible to find a performance bound for the model in \eqref{eq:hybrid_pdf} by means of the HCRLB \cite{Messer06,NoamM09_57}. Furthermore, this bound may, as we will see in the numerical section, be asymptotically attained by a hybrid ML/MAP estimator \cite{Yeredor00_48}. The following result adapted from \cite{Messer06} holds.
\begin{proposition}[Hybrid Cram\'er-Rao lower bound \cite{Messer06}]
Let $\hat{\paraminharm}$ be an unbiased estimator of $\paraminharm$ in the sense that $\expop_{\by,\inharm}( \hat{\theta}) = \theta$ and $\expop_{\by,\inharm}( \hat{\inharm} ) = \expop_{\inharm}( \inharm)$, for any $\theta$. Then,
\begin{align}
	\expop_{\by,\inharm}\left( (\hat{\paraminharm}-\paraminharm)(\hat{\paraminharm}-\paraminharm)^T  \right) \succeq \breve{\fim}(\paraminharm)^{-1}
\end{align}
where
\begin{align} \label{eq:hybrid_hcrlb_fim}
	\breve{\fim} = \expop_{\by,\inharm}\Big( \nabla_{\paraminharm}\log p(\by,\Delta;\theta) \: \nabla_{\paraminharm}\log p(\by,\Delta;\theta)^T  \Big).
\end{align}
Here, $\expop_{\by,\inharm}$ and $\expop_{\inharm}$ denote expectation with respect to the joint pdf of $\by$ and $\inharm$, and the marginal pdf of $\inharm$, respectively.
\end{proposition}
For non-linear measurement models, such as the one considered herein, the HCRLB is in general only tight asymptotically \cite{NoamM09_57}. However, it can be shown that if the bound is tight, then it is attained by the hybrid ML/MAP estimator \cite{NoamM09_57}. This property makes the HCRLB attractive, especially for the case considered herein, where it may be computed easily, as detailed in the following proposition. Furthermore, the ML/MAP allows for straightforward implementation.
\begin{proposition} \label{prop:hcrlb}
For the inharmonic pitch model in \eqref{eq:hybrid_pdf}, the matrix $\breve{\fim}$ in \eqref{eq:hybrid_hcrlb_fim} defining the HCRLB is given by 
\begin{align}
	\breve{\fim} = \expop_\inharm\left( \fim(\paraminharm)  \right) + \begin{bmatrix}
		0 & 0 \\
		0& \frac{1}{\inharmvar} I
	\end{bmatrix}
\end{align}
where
\begin{align*}
	\fim(\paraminharm) =\frac{2}{\truevar}\!\sum_{t=0}^{N-1}\!\nabla_{\paraminharm} \missignal_t^\rreal(\paraminharm) 	\nabla_{\paraminharm} \missignal_t^\rreal(\paraminharm)^T\!+\!\nabla_{\paraminharm} \missignal_t^\iimag(\paraminharm) \nabla_{\paraminharm} \missignal_t^\iimag(\paraminharm)^T
\end{align*}
and $I$ denoting the identity matrix of size $K\times K$. A detailed expression for $ \expop_\inharm( \fim(\paraminharm))$ may be found in the appendix. 
\end{proposition}
\begin{proof}
See the appendix.
\end{proof}
\begin{remark}
Partitioning $\breve{\fim}$ as
\begin{align*}
	\breve{\fim}= \begin{bmatrix}
		F_{\theta,\theta} & F_{\theta,\inharm}^T \\
		F_{\theta,\inharm} & F_{\inharm,\inharm}
	\end{bmatrix}
\end{align*}
where $F_{\theta,\theta}$ is the $(2K+1) \times (2K+1)$ block corresponding to the non-random parameters and $F_{\inharm,\inharm}$ corresponds to the $K$ inharmonicity parameters, it may be noted $F_{\theta,\theta}$ converges to the FIM corresponding to a perfectly harmonic model when letting $\inharmvar \to 0$. Also, $\frac{1}{N}F_{\inharm,\inharm} \approx \frac{1}{N}\frac{1}{\inharmvar} I $ when $\inharmvar \to 0$ and $N$ is reasonably large. 
Noting that the Schur complement of $F_{\inharm,\inharm}$ in $\breve{\fim}$ is $F_{\theta,\theta} -F_{\theta,\inharm}^T F_{\inharm,\inharm}^{-1}F_{\theta,\inharm}$, and noting that $F_{\theta,\inharm}$ converges to a finite matrix as $\inharmvar \to 0$, the top $(2K+1) \times (2K+1)$ block of $\breve{\fim}^{-1}$ converges to $F_{\theta,\theta}^{-1}$, as $F_{\theta,\inharm}^T F_{\inharm,\inharm}^{-1}F_{\theta,\inharm} \to \inharmvar F_{\theta,\inharm}^TF_{\theta,\inharm} \to 0$, when $\inharmvar \to 0$. That is, the HCRLB converges to the CRLB of the perfectly harmonic model as $\inharmvar\to 0$.
\end{remark}
Furthermore, it can be shown that, given some regularity conditions on the pdf $p$ \cite{NoamM09_57}, the HCRLB is asymptotically tight and asymptotically attained by the hybrid ML/MAP estimator, i.e., by $\text{arg max}_{\theta,\inharm}\: p(\by,\inharm;\theta)$. Thus, the bound in Prop.~\ref{prop:hcrlb} constitutes a useful predictor of estimation performance. Next, we present the ML/MAP estimator for the inharmonic pitch model and show that it lends itself to straightforward implementation. From \eqref{eq:hybrid_pdf}, it may be noted the ML/MAP estimate of $(\theta,\truevar,\Delta)$ maximizes the function
\begin{align*}
	\mathcal{L}= -N\log\truevar-\frac{1}{\truevar} \norm{\by-\bx(\paraminharm)}_2^2-\frac{1}{2\sigma^2_\Delta} \norm{\Delta}_2^2,
\end{align*}
which is the log-likelihood of \eqref{eq:hybrid_pdf}, excluding constant terms. It may be noted that we here are required to estimate also the noise variance $\truevar$. In order to formulate the ML/MAP estimator, define the dictionary function $A: \RR^K \to \RC^{N \times K} $
\begin{align*}
	A(\bomega) = \begin{bmatrix}	a(\omega_1) & \ldots &a(\omega_K) \end{bmatrix}
\end{align*}
where $\bomega = \left[ \begin{array}{ccc} \omega_1 & \ldots & \omega_K \end{array} \right]^T$ and $a: \RR \to \RC^N$ is the Fourier vector.
\begin{proposition}[ML/MAP estimator] \label{prop:map_estimator}
Let $\bomega$ be the set of frequencies maximizing the function
\begin{align*}
	\mapcriterion(\bomega) = -N\log \Sigma(\bomega) - \frac{1}{2\sigma^2_\Delta} \nu(\bomega),
\end{align*}
where
\begin{align*}
	\Sigma(\bomega) &=  \frac{1}{N}\norm{\by-A(\bomega) \left(A(\bomega)^HA(\bomega)  \right)^{-1}A(\bomega)^H\by}_2^2, \\
	\nu(\bomega) &= \sum_{k=1}^K \left( \omega_k - k\frac{\sum_{\ell=1}^K \ell \omega_\ell}{\sum_{\ell=1}^K \ell^2}  \right)^2.
\end{align*}
Then, the ML/MAP estimates of the fundamental frequency and inharmonicity parameters are given by
\begin{align*}
	\omega_0 = \frac{\sum_{k=1}^K k \omega_k}{\sum_{k=1}^K k^2}
\end{align*}
and $\inharm_\ell = \omega_\ell - \ell\omega_0$, for $\ell=1,\ldots,K$, respectively.
\end{proposition}
\begin{remark}
It may be noted that the ML/MAP estimate of the model parameters are found by maximizing $\mapcriterion$ over a set of $K$ unconstrained frequencies, which may be realized by a non-linear search. As $\mapcriterion$ is non-convex, such a search requires a good initial point. In favorable noise conditions, such an initial point may be obtained by simple peak-picking in the periodogram.
\end{remark}
%
%
\begin{figure}[t]
        \centering
            \includegraphics[width=.46\textwidth]{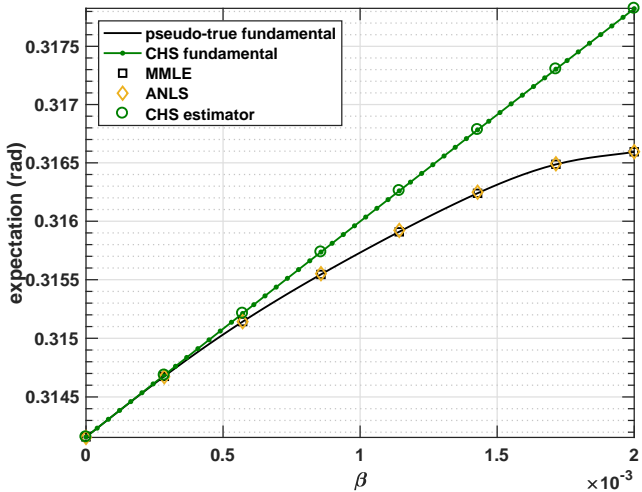}
           \caption{The pseudo-true fundamental frequency in Def.~\ref{def:l2_pitch} as well as the CHS fundamental in Def.~\ref{def:omt_spectrum} for signals generated from the string model in \eqref{eq:string_model} for varying $\beta$. Also plotted are the estimated expectations of the MMLE, ANLS, and CHS estimators.} 
            \label{fig:non_stoch_expectation_vary_beta}
\vspace{-2mm}\end{figure}
%
%
As can be seen from the criterion $\mapcriterion$, one may obtain two extreme cases by letting $\inharmvar \to \infty$ and $\inharmvar \to 0$, respectively. In the former case, maximizing $\mapcriterion$ becomes equivalent to minimizing $\Sigma$ with respect to a set of unconstrained frequencies $\bomega$, i.e., one obtains the MLE for a model with $K$ unrelated sinusoids. In the latter case, one in the limit arrives at the problem
\begin{align*}
	\maxwrt[\bomega] -\Sigma(\bomega) \;,\; \text{s.t. } \nu(\bomega) = 0,
\end{align*}
or, equivalently,
\begin{align*}
	\maxwrt[\omega_0,\bomega] -\Sigma(\bomega) \;,\; \text{s.t. } \omega_k= k\omega_0, \text{ for } k = 1,\ldots,K,
\end{align*}
i.e., the misspecified MLE corresponding to a perfectly harmonic approximation.
%
%
%
\section{Numerical examples} \label{sec:numerical}
In this section, we provide numerical examples illustrating the derived theoretical results.
%
\subsection{Deterministic waveform}
To illustrate the behavior of Def.~\ref{def:l2_pitch}, i.e.,  the best harmonic approximation in $\ell_2$ sense, and the implied MMLE for varying degrees of inharmonicity, we consider signals generated from the string model in \eqref{eq:string_model}. It may here be noted that the MMLE is given by the non-linear least squares (NLS) estimator from \cite{ChristensenSJJ08_88}. In particular, let the signal consist of $K = 5$ components with amplitudes $\trueamp_k = e^{-\rho(k-K/2)^2}$ with $\rho = 0.2$, and let $\pitch_0 = \pi/10$. The initial phases $\truephase_k$ are chosen uniformly random on $[-\pi,\pi)$. Furthermore, let the signal be observed at $N = 500$ time instances, and let the SNR, defined as $\text{SNR} = 10\log_{10} {\sum_{k} \trueamp_k^2}/\truevar$, be 10 dB. For this setting, Figure~\ref{fig:non_stoch_expectation_vary_beta} presents the pseudo-true pitch i.e., the $\pitch_0$ defined in Def.~\ref{def:l2_pitch}, when varying the string stiffness parameter $\beta$ on $[0,2\times 10^{-3}]$. Also presented is the OMT-based CHS $\omega_0$ in Def.~\ref{def:omt_spectrum}. As may be noted, both definitions correspond to small perturbations of the pitch, with the CHS definition displaying a more linear behavior for a larger range of $\beta$. Figure~\ref{fig:non_stoch_expectation_vary_beta} also displays the estimated expected values for the MMLE and CHS estimators, obtained from 2000 Monte Carlo simulations for each value of $\beta$. The CHS pitches are computed based on unstructured ML estimates of the component amplitudes and frequencies in accordance with Prop.~\ref{prop:chs_estimate}. As can be seen, the sample averages correspond well to their theoretical values. Also presented is the average estimate obtained using the approximate NLS (ANLS) estimator \cite{ChristensenSJJ08_88}, which is an asymptotic approximation of the MLE for an harmonic model (as $N\to \infty$), corresponding to harmonic summation from a periodogram estimate. It may here be noted that the expectation of the ANLS estimator coincides with the pseudo-true $\pitch_0$.
%
%
\begin{figure}[t]
        \centering
            \includegraphics[width=.45\textwidth]{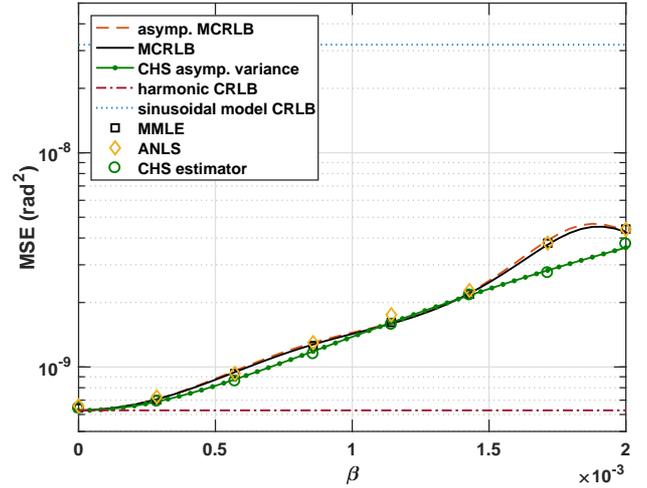}
           \caption{The MCRLB, the asymptotic MCRLB, and the asymptotic CHS estimator variance for estimating the pseudo-true $\pitch_0$ in Def.~\ref{def:l2_pitch} and the CHS pitch in Def.~\ref{def:omt_spectrum}, respectively, when the signal is generated from the string model in \eqref{eq:string_model} for varying $\beta$. Also plotted are the corresponding MSE achieved by the MMLE, ANLS, and CHS estimators.}
            \label{fig:non_stoch_mse_vary_beta}
\vspace{-2mm}\end{figure}
%
%

Furthermore, Figure~\ref{fig:non_stoch_mse_vary_beta} presents the exact and asymptotic MCRLB, as well the asymptotic variance of the CHS estimator, together with the MSE of the MMLE, ANLS, and CHS estimators. The MSE is here computed using the pseudo-true $\pitch_0$ as reference for the MMLE and ANLS, as it corresponds to their expected and asymptotical expected values, respectively, whereas the reference for the CHS estimator is the CHS $\pitch_0$. As may be noted, the bounds provide accurate predictions of the behaviors of the three estimators, with the asymptotic MCRLB coinciding fairly well with its exact counterpart. As reference, the CRLB for the corresponding perfectly harmonic model, as well as for the lowest-frequency component in an unstructured sinusoidal model with exactly the same spectral content, are also provided.
Figures~\ref{fig:non_stoch_expectation_vary_N} and \ref{fig:non_stoch_mse_vary_N} display corresponding quantities, i.e., theoretical and estimated expected values and theoretical and estimated variances, respectively, when fixing $\beta = 5 \times 10^{-4}$ and varying the number of samples, $N$, between 300 and 2000. As can be seen in Figures~\ref{fig:non_stoch_expectation_vary_N}, the CHS $\pitch_0$ does not depend on the sample length, as expected. In contrast, the pseudo-true $\pitch_0$ depends on $N$. As $\trueamp_2 = \trueamp_3$ dominate the amplitudes of the other components, we expect the $\pitch_0$ from Def.~\ref{def:l2_pitch} to fluctuate close to the interval $[\pitch_0\sqrt{1+\beta 2^2},\pitch_0\sqrt{1+\beta 3^2}] \approx [0.3145,0.3149]$ for large but moderate values of $N$ (recall that Def.~\ref{def:l2_pitch} becomes ambiguous as $N\to \infty$). Furthermore, as may be seen from Figure~\ref{fig:non_stoch_mse_vary_N}, the MCRLB is far from being linear in the log of $N$, in contrast to the CHS asymptotic variance. It may here be noted that the slopes for the perfectly harmonic CRLB and the CHS asymptotic variance are different as the second term of \eqref{eq:chs_var} dominates for large $N$.
%
%
\begin{figure}[t]
        \centering
            \includegraphics[width=.46\textwidth]{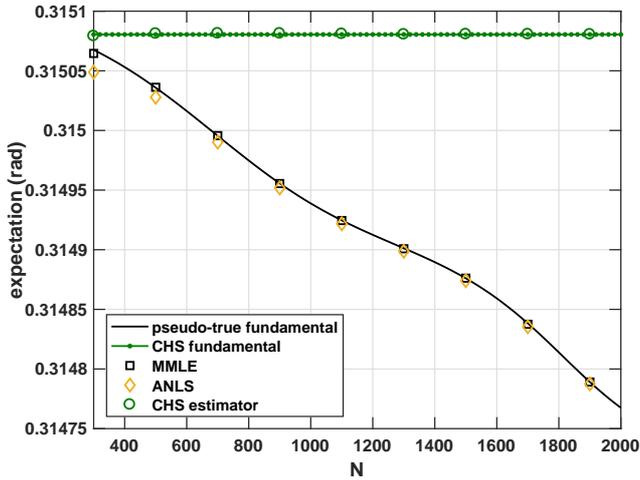}
           \caption{The pseudo-true fundamental frequency in Def.~\ref{def:l2_pitch} as well as the CHS fundamental in Def.~\ref{def:omt_spectrum} for signals generated from the string model in \eqref{eq:string_model} for varying $N$. Also plotted are the estimated expectations of the MMLE, ANLS, and CHS estimators.} 
            \label{fig:non_stoch_expectation_vary_N}
\vspace{-2mm}\end{figure}
%
%
\subsection{Stochastic waveform}
We proceed by extending the simulation study to the stochastic model in \eqref{eq:stochastic_harmonic_sines}. Specifically, we let all signal parameters remain the same, with the difference being that each component frequency is perturbed by a Gaussian random variable. Fixing $N = 500$ and SNR $= 10$ dB, we compute the HCRLB for $\pitch_0$, as well as for the (random) frequency of the first sinusoidal component, i.e., $\truefreq_1 = \omega_0 + \inharm_1$, for varying values of the inharmonicity variance $\inharmvar$. As comparison, we compute the theoretical MSE for the estimators implied by Defs.~\ref{def:l2_pitch} and \ref{def:omt_spectrum}. This is performed by sampling $\inharm_k \in \mathcal{N}\left(0,\inharmvar\right)$, computing the MSE for that particular set of $\left\{\inharm_k\right\}$, and averaging over realizations. Specifically, the MSE for a given $\left\{\inharm_k\right\}$ is computed by adding the squared bias, with references $\pitch_0$ and $\pitch_0 + \inharm_1$, to the covariance bound, i.e., the MCRLB and the CHS asymptotic variance, respectively. For each value of $\inharmvar$, the corresponding MSEs are computed by averaging over 1000 Monte Carlo simulations. Furthermore, for each such simulation, noise is added to the signal waveform according to \eqref{eq:measurement_equation}, and the signal parameters are estimated using MAP/MLE, MMLE, ANLS, and the CHS estimators. The results are displayed in Figures~\ref{fig:stoch_mse_first_omega0_sigma2delta} and \ref{fig:stoch_mse_first_sine_sigma2delta}, with Figure~\ref{fig:stoch_mse_first_omega0_sigma2delta} showing the MSE for the estimate of $\pitch_0$ and Figure~\ref{fig:stoch_mse_first_sine_sigma2delta} the MSE for $\truefreq_1 = \pitch_0 + \inharm_1$.
As may be seen from Figure~\ref{fig:stoch_mse_first_omega0_sigma2delta}, the MAP/MLE estimator is able to attain the HCRLB, which is strictly smaller than the bounds for the estimators derived for the deterministic models.
For reference, Figure~\ref{fig:stoch_mse_first_omega0_sigma2delta} provides also the MCRLB, the theoretical CHS asymptotic variance, as well as the obtained MSEs for estimates of the $\pitch_0$ of Defs.~\ref{def:l2_pitch} and \ref{def:omt_spectrum}. From this, it may be concluded that the variances of the estimators are dominated by their squared bias when estimating $\pitch_0$ from the model in \eqref{eq:stochastic_harmonic_sines}.
Furthermore, it may be noted from Figure~\ref{fig:stoch_mse_first_sine_sigma2delta} that when considering the estimate of the frequency of the first sinusoidal component, the HCRLB, as well as the MSE of the MAP/MLE, tends to the CRLB corresponding to an unstructured sinusoidal model as $\inharmvar$ grows, which is due to negative correlation between the estimates of $\pitch_0$ and $\inharm_1$. In contrast, for large enough inharmonicity, the MSEs of the MMLE, ANLS, and CHS estimators exceed the CRLB of the sinusoidal model.
%
%
%
\begin{figure}[t]
        \centering
            \includegraphics[width=.46\textwidth]{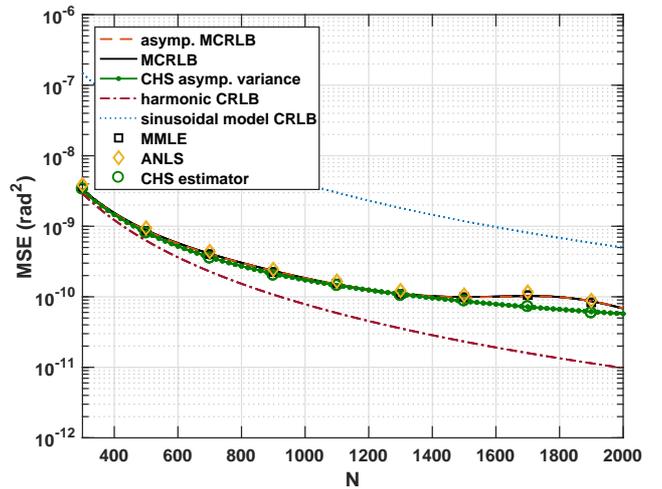}
           \caption{The MCRLB, the asymptotic MCRLB, and the asymptotic CHS estimator variance for estimating the pseudo-true $\pitch_0$ in Def.~\ref{def:l2_pitch} and the CHS pitch in Def.~\ref{def:omt_spectrum}, respectively, when the signal is generated from the string model in \eqref{eq:string_model} for varying $N$. Also plotted are the corresponding MSE achieved by the MMLE, ANLS, and CHS estimators.}
            \label{fig:non_stoch_mse_vary_N}
\vspace{-2mm}\end{figure}
%
%
\vspace{-2mm}
\section{Discussion}\label{sec:discussion}
The three different definitions of the pitch for inharmonic signals all have their merits, and one cannot in a strictly objective sense say that one is better than the others. Instead, their relative usefulness depend on the considered application, as well as on ones aim with approximating close-to-harmonic signals with perfectly harmonic counterparts.

As noted, by defining pitch using an approximation in $\ell_2$, as in Def.~\ref{def:l2_pitch}, one models the scenario of erroneously assuming that the observed waveform is periodic, for which Def.~\ref{def:l2_pitch} provides means for analyzing the average behavior of the misspecified MLE, as well as a bound on the estimation performance. Thus, Def.~\ref{def:l2_pitch} is well suited for being used as a benchmark when analyzing the behavior of pitch estimators when applied to inharmonic signals; the definition makes it possible to compute the bias and MSE, at least empirically, for any such estimator. However, the behavior of Def.~\ref{def:l2_pitch} is, as demonstrated in the numerical examples, non-linear with respect to both the inharmonicity and sample length, and is in addition ambiguous for very long signals.

In this respect, the OMT-based approximation of Def.~\ref{def:omt_spectrum} has the appealing quality of not depending on the actual signal, but only on its spectral properties. As shown both in the theoretical and numerical results, the behavior of Def.~\ref{def:omt_spectrum} is locally linear with respect to the signal inharmonicity, and corresponds well to the intuitive idea that inharmonicity corresponds to perturbations in frequency of spectral peaks. In addition to this, when used for estimation, Def.~\ref{def:omt_spectrum} displays linear behavior in terms of estimator variance, and coincides with the MLE for the perfectly harmonic case. With this in mind, Def.~\ref{def:omt_spectrum} is more satisfactory than Def.~\ref{def:l2_pitch} when considered a tool for understanding inharmonic signals. Furthermore, its connection to the EXIP framework, i.e., the idea of fitting given parameter estimates to a certain structure, makes it relevant as a benchmark also in practical estimation scenarios.
%
%
\begin{figure}[t]
        \centering
            \includegraphics[width=.46\textwidth]{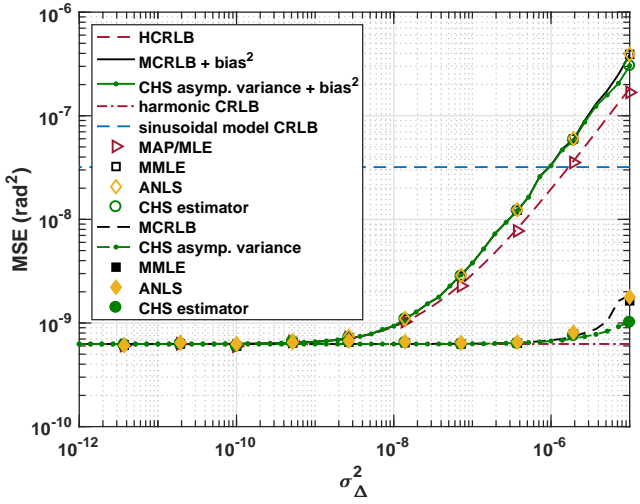}
           \caption{The MSE obtained when estimating $\pitch_0$ from the model in \eqref{eq:stochastic_harmonic_sines} for varying values of the inharmonicity parameter variance $\inharmvar$. Also given are MSEs for the estimates of the fundamental frequencies in Defs.~\ref{def:l2_pitch} and \ref{def:omt_spectrum}.} 
            \label{fig:stoch_mse_first_omega0_sigma2delta}
\vspace{-2mm}\end{figure}
%

In contrast to Defs.~\ref{def:l2_pitch} and \ref{def:omt_spectrum}, which constitute approximations, Def.~\ref{def:stochastic_def} views the inharmonic signal as a waveform resulting from perturbing the frequencies of the sinusoidal components by zero-mean random variables, causing them to deviate from perfect integer multiples of the nominal pitch. In this respect, Def.~\ref{def:stochastic_def} offers an explanation of why the observed signal waveform is not periodic. However, as any given observation of the signal has been generated by this random procedure, one by definition cannot compute the $\pitch_0$ in Def.~\ref{def:stochastic_def} from such an observation. However, adopting this view of inharmonic signals allows for computing a performance bound constituting a smooth interpolant between perfectly harmonic models and completely unstructured sinusoidal models, as well as for finding an easily implemented estimator attaining the bound. In practical terms, Def.~\ref{def:stochastic_def} offers tools for processing inharmonic signals for which one suspects that there may be no deterministic description of the inharmonic deviations. For example, such a class of signals could be the voiced part of human speech, for which one may observe frequency-dependent inharmonicity with no particular structure (this may very well change over time for a given person due to, e.g., infections). As the inharmonicity pattern thus could change from day to day for any given fundamental frequency, Def.~\ref{def:stochastic_def} likely provides a pragmatic view.
%
\begin{figure}[t]
        \centering
            \includegraphics[width=.46\textwidth]{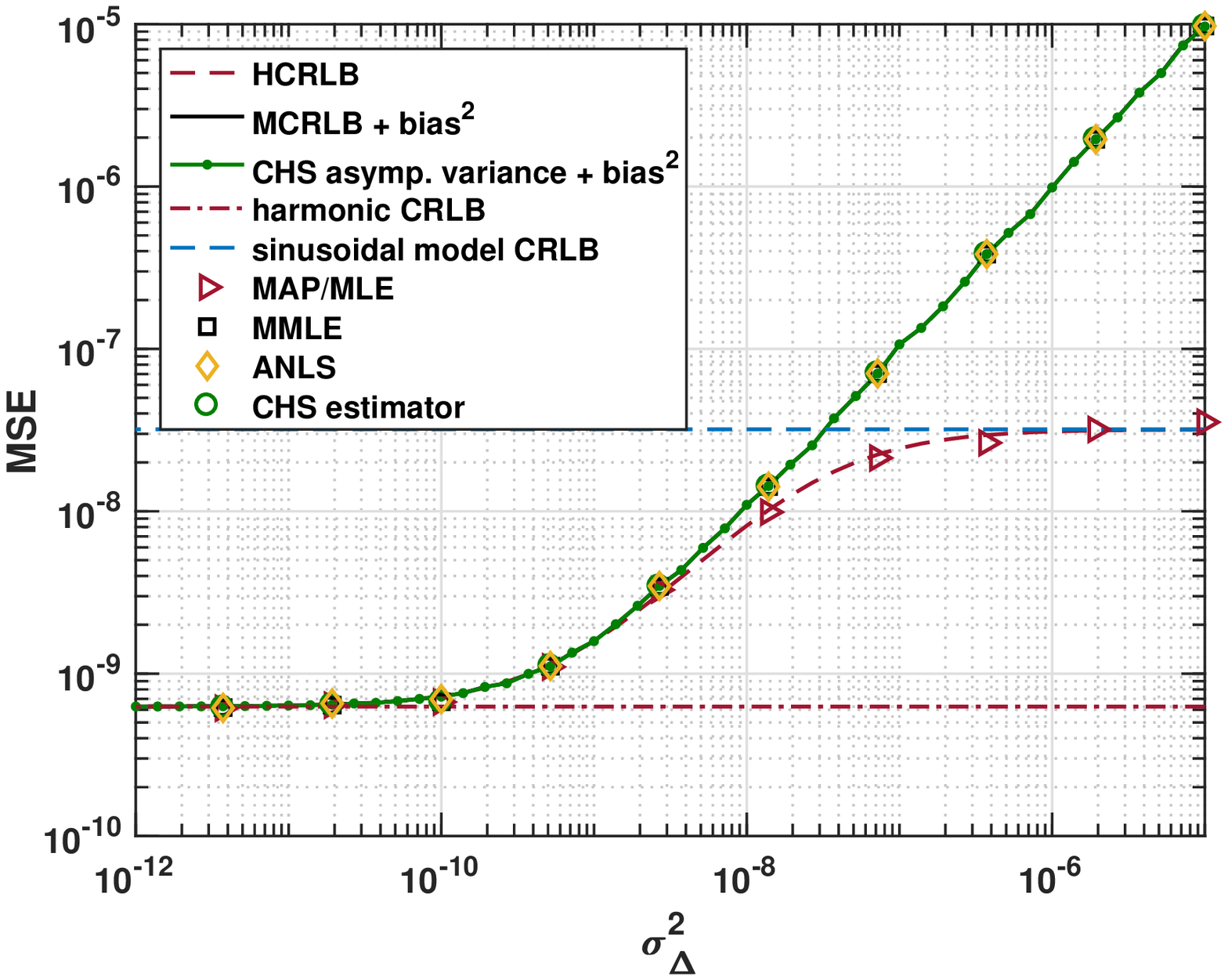}
           \caption{The MSE obtained when estimating the frequency $\truefreq_1 = \pitch_0 + \inharm_1$ of the first sinusoidal component from the model in \eqref{eq:stochastic_harmonic_sines} for varying values of the inharmonicity parameter variance $\inharmvar$.} 
            \label{fig:stoch_mse_first_sine_sigma2delta}
\vspace{-2mm}\end{figure}
%
\newpage
\appendix
\subsection{Information matrix for the HCRLB}
As $\expop_{\inharm_k}\left( e^{i \inharm_k t} \right) = e^{-\frac{1}{2}\inharmvar t^2}$, one may write $\expop_\inharm\left( \fim(\paraminharm)  \right) = \sum_{t=0}^{N-1} \Lambda^{(t)}$, where the elements of the matrices
\begin{align*}
	\Lambda^{(t)} \triangleq \expop_\inharm\left( \!\nabla_{\paraminharm} \missignal_t^\rreal(\paraminharm) 	\nabla_{\paraminharm} \missignal_t^\rreal(\paraminharm)^T\!+\!\nabla_{\paraminharm} \missignal_t^\iimag(\paraminharm) \nabla_{\paraminharm} \missignal_t^\iimag(\paraminharm)^T\right),
\end{align*}
for $t = 0,\ldots,N-1$, may be expressed as
\begin{align*}
	\Lambda^{(t)}_{\omega_0,\omega_0} \!&=\!  t^2 \!\sum_{k}\!k^2\trueamp_k^2 \\
	&\!+\!t^2e^{-\inharmvar t^2}\!\sum_{(k,\ell): k\neq \ell}\!k\ell \trueamp_k \trueamp_\ell\cos\left( \!\pitch(k-\ell)t+\!\truephase_k\!-\!\truephase_\ell \! \right) \\
	\Lambda^{(t)}_{\omega_0,r_k} &=  t e^{-\inharmvar t^2}\sum_{\ell: \ell\neq k} \ell\trueamp_\ell\sin\left( \pitch(k-\ell)t +\truephase_k-\truephase_\ell  \right)\\
	\Lambda^{(t)}_{\omega_0,\phi_k}\!&=\!t k\trueamp_k^2\!+\!te^{-\inharmvar t^2}\!\sum_{\ell: \ell\neq k}\! \ell \trueamp_k \trueamp_\ell \!\cos\!\left( \!\pitch(k\!-\!\ell)t\!+\!\truephase_k\!-\!\truephase_\ell \!\right) \\
	\Lambda^{(t)}_{r_k,r_\ell} &= \begin{cases} 1, & \ell = k\\ e^{-\inharmvar t^2} \cos\left( \pitch(k-\ell)t+\truephase_k-\truephase_\ell  \right), &\ell \neq k \end{cases} \\
	\Lambda^{(t)}_{r_k,\phi_\ell} &= \begin{cases} 0, & \ell = k\\ e^{-\inharmvar t^2} \sin\left( \pitch(k-\ell)t+\truephase_k-\truephase_\ell  \right), &\ell \neq k \end{cases} \\
	\Lambda^{(t)}_{\phi_k,\phi_\ell} &= \begin{cases} \trueamp_k^2, & \ell = k\\ e^{-\inharmvar t^2}\trueamp_k\trueamp_\ell \cos\!\left( \!\pitch(k-\ell)t\!+\!\truephase_k\!-\!\truephase_\ell  \right), &\ell\neq k \end{cases}
\end{align*}
and \mbox{$\Lambda^{(t)}_{\eta,\inharm_k} \!=\! t\Lambda^{(t)}_{\eta,\phi_k}$, for $\eta \!\in\! \{\! \omega_0,\! r_k,\! \phi_k\!\}$, and \mbox{$\Lambda^{(t)}_{\inharm_k,\inharm_\ell} \!=\! t^2 \!\Lambda^{(t)}_{\phi_k,\phi_\ell}$}.}
\subsection{Proof of Proposition~\ref{prop:asymp_mcrlb}}
By Lemma~\ref{eq:lemma_arrowhead} below, $\frac{1}{N} \fim(\misparam_0)$ and $\frac{1}{N} A(\misparam_0)$ converge to arrowhead matrices. As shorthands, let $\fim = \fim(\misparam_0)$ and $A = A(\misparam_0)$. Then,
\begin{align*}
	A = -\frac{1}{\var} \begin{bmatrix}
		\firstelement & z^T \\
		z & \text{diag}(d)
	\end{bmatrix} + \mathcal{O}(1),
\end{align*}
where $\mathcal{O}(1)$ denotes a bounded matrix and where $\text{diag}(d)$ denotes the diagonal matrix with $d$ as its main diagonal, $\firstelement = \firstelement_\missignal+ \firstelement_\wavediff$ and $z = z_\missignal + z_\wavediff$ with
\begin{align*}
	&\firstelement_\missignal= 2\sum_{t=0}^{N-1} \abs{\frac{\partial \missignal_t}{\partial \pitch_0}}^2 \;,\; \firstelement_\wavediff = 2\sum_{t=0}^{N-1} \real\left( \overline{\wavediff_t} \frac{\partial^2 \missignal_t}{\partial \pitch_0^2}\right), \\
	&d = 2\sum_{t=0}^{N-1}\nabla_\misscale \missignal_t^\rreal \odot \nabla_\misscale \missignal_t^\rreal +\nabla_\misscale \missignal_t^\iimag \odot \nabla_\misscale \missignal_t^\iimag, \\
	&z_\missignal= 2\sum_{t=0}^{N-1}\real\left(\overline{\frac{\partial \missignal_t}{\partial \pitch_0}} \nabla_\misscale \missignal_t \right) \:,\: z_\wavediff  \!= \!2\sum_{t=0}^{N-1} \real\left(\overline{\wavediff_t} \frac{\partial}{\partial \pitch_0}\nabla_\misscale \missignal_t \right),
\end{align*}
where $\odot$ denotes the Hadamard product, and
\[
\alpha = \left[ \begin{array}{cccccc} \phase_1 & \ldots & \phase_K & \amp_1 & \ldots & \amp_K \end{array} \right]^T,
\]
with all derivatives being evaluated at $\misparam = \misparam_0$. Then, by the Sherman-Morrison-Woodbury formula \cite{GolubV96}, $A^{-1}$ may be written as
\begin{align*}
	A^{-1} = -\var\begin{bmatrix}
		0 & 0 \\ 0 & \text{diag}(d)^{-1}
	\end{bmatrix} - \frac{\var}{\rho}uu^T + \mathcal{E},
\end{align*}
where $u = u_\missignal + u_\wavediff$, with $\rho = \firstelement - z^T(z./ d)$ and
\begin{align*}
	 u_\missignal = \left[ \begin{array}{cc} -1 & (z_\missignal./ d)^T \end{array} \right]^T \:,\: u_\wavediff = \left[ \begin{array}{cc} 0 & (z_\wavediff./ d)^T \end{array} \right]^T,
\end{align*}
where $./$ denotes elementwise division, and where the error matrix $\mathcal{E}$ is structured as
\begin{align*}
	\mathcal{E} = \begin{bmatrix}
		\mathcal{E}_{1} & \mathcal{E}_2^T \\
		\mathcal{E}_2 & \mathcal{E}_3
	\end{bmatrix}
\end{align*}
where $\mathcal{E}_{1}$ is a scalar on the order $\mathcal{O}(N^{-4})$, $\mathcal{E}_{2}$ is a $2K$ vector on the order $\mathcal{O}(N^{-3})$, and $\mathcal{E}_{3}$ is a $2K\times 2K$ matrix on the order $\mathcal{O}(N^{-2})$. Thus, as $\fim$ is given by
\begin{align*}
	\fim = \frac{\truevar}{(\var)^2}\begin{bmatrix}
		\firstelement_\missignal & z_\missignal^T \\
		z_\missignal & \text{diag}(d)
	\end{bmatrix} + \mathcal{O}(1),
\end{align*}
straightforward calculations yield that the error incurred by neglecting the bounded terms of $A$ and $\fim$ may be bounded from above by a matrix tending to zero faster than $A^{-1}\fim A^{-1}$ by a factor $1/N$. Then, the MCRLB corresponding to $\pitch_0$ is given by the first diagonal element of $A^{-1}\fim A^{-1}$, which for large $N$ is given by
\begin{align*}
	\frac{(\var)^2}{\rho^2} u^TFu &= (\var)^2\frac{1}{\rho^2} \left( u_\missignal^TFu_\missignal + u_\wavediff^TFu_\wavediff \right) \\
	&= \var \frac{C + E}{\left( C - E + Z + D \right)^2},
\end{align*}
where it used that $u_\missignal^T \fim u_\wavediff = 0$, and
\begin{align*}
	C &= \firstelement_\missignal - z_\missignal^T(z_\missignal./d) \;,\; E = z_\wavediff^T(z_\wavediff./d),\\
	D &= -2z_\missignal^T(z_\wavediff./d) \;,\; Z = \firstelement_\wavediff.
\end{align*}
Assuming that the pseudo-true fundamental frequency is not too close to zero, the correlation between signal components corresponding to different harmonic orders tends to zero as $N\to \infty$. The asymptotic expressions for $C, E, D, $ and $Z$ stated in the proposition follow directly, together with the approximation error on the order $\mathcal{O}(N^{-4})$ for large $N$.
\hfill $\square$
%
\begin{lemma} \label{eq:lemma_arrowhead}
As $N\to \infty$, $\frac{1}{N} \fim(\misparam_0)$ and $\frac{1}{N} A(\misparam_0)$ converge to arrowhead matrices.
\end{lemma}
%
\begin{proof}
Firstly, it may be noted that as $\misparam_0$ solves the LS criterion in \eqref{eq:ls_approximation}, it directly follows that
\begin{align*}
	\sum_{t=0}^{N-1} \wavediff_t^\rreal \nabla_\misparam \missignal_t^\rreal + \sum_{t=0}^{N-1} \wavediff_t^\iimag \nabla_\misparam \missignal_t^\iimag = 0.
\end{align*}
Then, as any second derivative of $\missignal_t^\rreal$ and $\missignal_t^\iimag$ not involving differentiation with respect to $\pitch$ is equal to a common constant real scaling of elements of $\nabla_\misparam \missignal_t^\rreal$ and $\nabla_\misparam \missignal_t^\iimag$, respectively, only the first column and first row of $\extrahessian(\misparam_0)$ are non-zero.
It is straightforward to show that off-diagonal elements of $\fim(\misparam_0)$ not related to partial derivatives of $\pitch_0$ converge linearly to zero when scaled by $1/N$, whereas the diagonal is bounded from below by positive values. Thus, $\frac{1}{N} \fim(\misparam_0)$, and thereby $\frac{1}{N}A(\misparam_0)$, converges to an arrowhead matrix as $N \to \infty$.
\end{proof}
%
%
%
\subsection{Proof of Proposition~\ref{prop:bounded_perturbation}}
Clearly, if there are at least two consecutive sinusoids with non-zero amplitude, $d$ in Def.~\ref{def:max_harm_order} satisfies \mbox{$d \in \left[\tilde{\pitch}_0-2\norm{\inharm}_\infty,\tilde{\pitch}_0+2\norm{\inharm}_\infty\right]$}. Then,
\begin{align*}
	(K+1)d \geq \truefreq_K \;,\; (K-1)d < \truefreq_K,
\end{align*}
if $\norm{\inharm}_\infty < \tilde{\pitch}_0/(2K+3)$. Thus for $\norm{\inharm}_\infty < \tilde{\pitch}_0/(2K+3)$, we have that $L \in \left\{ K, K+1 \right\}$. Furthermore, for any $\pitch \in \left[\tilde{\pitch}_0 - \norm{\inharm}_\infty , \tilde{\pitch}_0 + \norm{\inharm}_\infty\right]$,
\begin{align*}
	\abs{k\pitch - \truefreq_k} \leq (k+1)\norm{\inharm}_\infty < \tilde{\pitch}_0 \frac{k+1}{2K+3},
\end{align*}
whereas $\abs{(k\pm1)\pitch - \truefreq_k} > \tilde{\pitch}_0 (2K+1-k)/(2K+3)$. Thus, $\argminwrt[\ell] (\ell\pitch-\truefreq_k)^2 = k$, for $k =1,\ldots,K$, implying
\begin{align*}
	q_{L}(\pitch) = 2\pi\sum_{k=1}^K \trueamp_k^2 (k\pitch - \truefreq_k)^2 = 2\pi\sum_{k=1}^K \trueamp_k^2 (k(\pitch-\tilde{\pitch}_0) - \inharm_k)^2
\end{align*}
for any $\pitch \in \left[\tilde{\pitch}_0 - \norm{\inharm}_\infty , \tilde{\pitch}_0 + \norm{\inharm}_\infty\right]$. This quadratic function has the unique stationary point
\begin{align*}
	\pitch_0 = \frac{\sum_{k=1}^K \trueamp_k^2 k \truefreq_k}{\sum_{k=1}^K \trueamp_k^2 k^2} = \tilde{\pitch}_0 + \frac{\sum_{k=1}^K \trueamp_k^2 k \inharm_k}{\sum_{k=1}^K \trueamp_k^2 k^2},
\end{align*}
where it may be noted that
\begin{align*}
	\abs{\frac{\sum_{k=1}^K \trueamp_k^2 k \inharm_k}{\sum_{k=1}^K \trueamp_k^2 k^2}} \leq \frac{\sum_{k=1}^K\trueamp_k^2 k}{\sum_{k=1}^K\trueamp_k^2 k^2} \norm{\inharm}_\infty \leq \norm{\inharm}_\infty.
\end{align*}
\hfill $\square$
\subsection{Proof of Proposition~\ref{prop:chs_estimate}}
We here assume that the inharmonic perturbations are small so that, asymptotically, i.e., when $N \to \infty$ or the SNR tending to infinity, the assumptions of Prop.~\ref{prop:bounded_perturbation} hold in the sense that there exist $\omega$ such that $|k\omega - \hat{\truefreq}_k | < \omega/(2K+3)$, for $k=1,\ldots,K$, almost surely. Noting that the covariance matrix of the vector $\hat{\theta} = \left[\begin{array}{cccccc} \hat{\trueamp}_1 & \ldots & \hat{\trueamp}_K & \hat{\truefreq}_1& \ldots & \hat{\truefreq}_K     \end{array}\right]^T$ is asymptotically given by (see, e.g., \cite{StoicaJL97_45})
\begin{align*}
	\text{Cov}(\hat{\theta}) \!=\! \frac{\truevar}{2N} \begin{bmatrix}   I& 0\\ 0 & C_2 \end{bmatrix} \:,\: C_2 \!=\! \frac{12}{N^2\!-\!1} \text{diag}\left( \left[ 1/\trueamp_1^2 \: \ldots \: 1/\trueamp_K^2\right]\right)
\end{align*}
and that the estimate of the CHS fundamental frequency is given $\hat{\pitch}_0 = f(\hat{\theta})$, where
\begin{align*}
	f(\theta) = \frac{\sum_{k=1}^K \trueamp_k^2 k \truefreq_k}{\sum_{k=1}^K \trueamp_k^2 k^2},
\end{align*}
the expression in \eqref{eq:chs_var} is the obtained from a first order Taylor expansion of $f$ at $\theta$. As can be seen from $\text{Cov}(\hat{\theta})$, the estimates of $\truefreq_k$ are asymptotically uncorrelated with estimates of the amplitudes, from which it directly follows that $\hat{\pitch}_0$ is an asymptotically unbiased estimate of $\pitch_0$. As $\hat{\theta}$ is the MLE of $\theta$, consistency follows, with the asymptotic variance being
\begin{align*}
	\expop\left( (\omega_0 - \hat{\omega}_0)^2 \right) = \nabla_{\theta} f(\theta)^T\text{Cov}(\hat{\theta})\nabla_{\theta} f(\theta).
\end{align*}
After some simplification, the expression in \eqref{eq:chs_var} follows.
\hfill $\square$
\subsection{Proof of Proposition~\ref{prop:hcrlb}}
We have
\begin{align*}
	\nabla_{\paraminharm}\log p(\by,\inharm;\theta) = \nabla_{\paraminharm}\log p(\by\mid\inharm;\theta) + \nabla_{\paraminharm}\log p(\inharm),
\end{align*}
where $\nabla_{\paraminharm}\log p(\inharm) = \text{diag}\left( \begin{array}{cc} 0^T & -\inharm^T/\inharmvar \end{array} \right)$, where $0$ is a zero vector of length $2K+1$. Further, as 
\begin{align*}
	\expop_{\by\mid\inharm}\left(\nabla_{\paraminharm}\log p(\by\mid\inharm;\theta) \right) = 0,
\end{align*}
and $\inharm$ is independent of the measurement noise,
\begin{align*}
	\expop_{\by\mid \inharm}\left( \nabla_{\paraminharm}\log p\; \nabla_{\paraminharm}\log p^T\right) = \fim(\paraminharm) + \begin{bmatrix} 0 & 0& \\0 & \frac{1}{(\inharmvar)^2}\inharm\inharm^T \end{bmatrix},
\end{align*}
where we use the shorthand $p = p(\by,\inharm;\theta)$, and 
\begin{align*}
	\fim(\paraminharm) &= \expop_{\by\mid\inharm}\left(\nabla_{\paraminharm}\log p(\by\mid\inharm;\theta) \nabla_{\paraminharm}\log p(\by\mid\inharm;\theta)^T \right) \\
	&= \frac{2}{\truevar}\!\sum_{t=0}^{N-1}\!\nabla_{\paraminharm} \missignal_t^\rreal(\paraminharm) 	\nabla_{\paraminharm} \missignal_t^\rreal(\paraminharm)^T\!+\!\nabla_{\paraminharm} \missignal_t^\iimag(\paraminharm) \nabla_{\paraminharm} \missignal_t^\iimag(\paraminharm)^T
\end{align*}
as the measurement noise is circularly symmetric white Gaussian. As $\expop_{\inharm}\left( \inharm\inharm^T \right) = \inharmvar I$, the expression for $\breve{\fim}$ follows directly. It may also be readily verified that
\begin{align*}
	 \expop_{\by\mid \inharm}\left(\nabla_{\paraminharm}\log p(\by\mid\inharm;\theta) \frac{\partial}{\partial \truevar} \log p(\by\mid\inharm;\theta)\right) = 0,
\end{align*}
i.e., the HCRLB for $\paraminharm$ does not depend on whether $\truevar$ is known or not, implying that no partial derivatives with respect to $\truevar$ need to be considered to compute the HCRLB of $\paraminharm$.
\hfill $\square$
\vspace{-2mm}
%
\bibliographystyle{IEEEbib}
\bibliography{ElvanderJ20_final_manuscript_TSP_arXiv,IEEEabrv}

\end{document}